\documentclass[%
 aip,
 amsmath,amssymb,
]{revtex4-1}
\usepackage{amsthm}
\usepackage{graphicx}
\usepackage{bm}

\usepackage[utf8]{inputenc}
\usepackage[T1]{fontenc}
\usepackage{mathptmx}
\usepackage{etoolbox}
\usepackage[normalem]{ulem} 

\usepackage[shortlabels]{enumitem}

\usepackage{amsfonts}
\usepackage[utf8]{inputenc}
\usepackage[T1]{fontenc}
\usepackage{mathptmx}
\newtheorem{definition}{Definition}[section]
\newtheorem{theorem}{Theorem}[section]

\newtheorem{remark}{Remark}[section]
\newtheorem{example}{Example}[section]

\newtheorem{proposition}{Proposition}[section]

\newtheorem{corollary}{Corollary}

\usepackage{color}

\usepackage{tikz}
\usetikzlibrary{matrix}

\makeatletter
\def\@email#1#2{%
 \endgroup
 \patchcmd{\titleblock@produce}
  {\frontmatter@RRAPformat}
  {\frontmatter@RRAPformat{\produce@RRAP{*#1\href{mailto:#2}{#2}}}\frontmatter@RRAPformat}
  {}{}
}%
\makeatother
\begin{document}

\preprint{AIP/123-QED}

\title{The generalized method of separation of variables for diffusion-influenced reactions: Irreducible Cartesian tensors technique} 


\author{Sergey D. Traytak}
\email{sergtray@mail.ru}
\affiliation{Semenov Federal Research Center for Chemical Physics, Russian Academy of Sciences, 4 Kosygina St., 119991 Moscow, Russian Federation}


\date{\today}

\begin{abstract}
Motivated through various applications of the trapping diffusion-influenced reactions theory in physics, chemistry and biology, this paper deals with irreducible Cartesian tensors (ICT) technique within the scope of the generalized method of separation of variables (GMSV). Presenting a survey from the basic concepts of the theory, we spotlight the distinctive features of the above approach against known in literature similar techniques. The classical solution to the stationary diffusion equation under appropriate boundary conditions is represented as a series in the ICT. By means of proved translation addition theorem we straightforwardly reduce the general boundary value diffusion problem for $N$ spherical sinks to the corresponding resolving infinite system of linear algebraic equations with respect to unknown tensor coefficients. These coefficients comprise explicit dependence on the arbitrary three-dimensional configurations of $N$ sinks with different radii and surface reactivities. Specific application of the ICT technique is illustrated by some numerical calculations for the reactions, occurring in the arrays of two and three spherical sinks. For the first time we treat the statement and solution to the above reaction-diffusion problem, applying the smooth manifold concept.
\end{abstract}

\maketitle


%
%
%


%
%
\section{Motivation}\label{sec:motivation}

It is presently accepted that for the most part, chemical reactions in micro-heterogeneous liquid media are {\it contact diffusion-influenced}~\cite{Eyring80,Rice85,Ovchinnikov89,Barzykin01,Dagdug24}. The latter means that the rate of these reactions is determined significantly by the rate of encounter of reactants due to diffusion; that is, the reaction rate is not controlled by only the chemical requirement of overcoming an activation energy barrier. 

Due to their high abundance, diffusion-influenced reactions play a decisive role for a wide and diversified range of applications often occurring in physics, chemistry, biology and nanotechnology. Examples include excitation quenching of donors by acceptors, heterogeneous catalysis, crystal defect annealing, crystal growth, Ostwald ripening, evaporating or burning of liquid droplets, nutrient consumption by living cells and cell metabolism to name just a few.

To describe microphysics of diffusion-influenced reactions theoretically often so-called {\it Smoluchowski's trapping model} is used, i.e., assuming that particles diffuse in a continuum medium containing immobile uniformly distributed absorbing sinks. Smoluchowski's theory is essentially a one-sink theory, which does not account for the influence of neighboring sinks. Quite apparently the presence of other neighboring absorbing sinks decreases the local concentration of diffusing particles, and so each sink captures less particles than it would a single isolated. This, obviously, implies that many-sink effects should be incorporated into theoretical analysis. According to accepted terminology, many-sink effects are called the diffusive interaction (for an extended discussion see Sec.~\ref{sec:DiffInter}). Corresponding mathematical problems are well-defined, however, the theoretical description of the diffusive interaction is mostly so complicated that solving of such problems analytically at least approximately is very difficult if it is even possible. 

We shall investigate here micro-heterogeneous medium in which spherical obstacles including absorbing sinks, of different radii are distributed in a second distinct continuous phase. The dispersed obstacles assumed to be non-overlapping and impenetrable.

It has been shown explicitly that the diffusive interaction effects become the most profound in the steady state~\cite{Traytak95,Traytak08} and it is this regime that will be studied in this work. A number of different analytical approaches proposed earlier to solve the posed diffusion problems for the steady state were reviewed in Ref.~\cite{Traytak03}. The most complete list of the commonly used analytical and numerical methods to solve the posed problems along with some comments was given later in Refs.~\cite{Traytak18,Grebenkov19}.

Here, we consider a modification of the GMSV that goes back to Rayleigh's seminal paper on the conductivity of heat and electricity in a medium comprising cylindrical or spherical sink systems~\cite{Rayleigh1892} (see Sec.~\ref{sec:method}). Although the GMSV by means of ICT seems to be fairly convenient at least for numerical calculations, however it still did not draw due proper attention to researchers. The work extends our previous studies via the GMSV~\cite{Traytak91,Traytak92,Traytak03,Traytak18,Grebenkov19} to the general case of the steady-state diffusion-influenced processes. We focus on the development a detailed theory of the diffusive interaction effects (see Sec.~\ref{sec:DiffInter}) by means of the ICT technique within the framework of microscopic Smoluchowski's trapping model. It turns out that  solution in terms of the ICT elucidates rather subtle mathematical facets of the multipole method and, what is more important, it significantly simplifies the calculation of the diffusion field and corresponding reaction rates. 

This research contains all the pertinent mathematical details including terminology, definitions, geometrical structure, convergence and step by step description of the GMSV algorithm with the ICT technique to solve the general diffusion boundary value problem within the scope of Smoluchowski's trapping model.

Therefore, the main purpose of this research is twofold. First, we intend to pose rigorously and in full details the boundary value problem for the most general case of boundary conditions, describing the small particles diffusion and absorption in a medium comprising both static inert obstacles and reactive sinks. Secondly, providing the overall view onto the GMSV, we solve the posed problem by means of the ICT technique, keeping in mind its practical aspects. Besides, since reviews about the problem are lacking in the existing literature, so it is our intention to provide a rather comprehensive overview of the current research activities in the field. Immediately, note here that, it seems expedient to give brief literature surveys on the subjects under consideration in the appropriate places of the text.

Thus, we hope that, first of all, this paper will repay the attention of mathematically inclined theoreticians working in a wide range of the applied scientific fields dealing with many-body effects that occur in diffusion-influenced reactions and other related topics.

\section{Introduction}\label{sec:intro}

To provide a self-contained presentation, we start with a brief survey of the microscopic Smoluchowski trapping theory of the diffusion-influenced reactions. Without going into details, we particular emphasis on the main mathematical ideas in order to generalize easily its results to the most general case studied in this paper.

The standard notation: $\mathbb{N}$, $ \mathbb{Z} $, $\mathbb{R}$ and $\mathbb{R}_+$ the sets of natural, integer, real and strictly positive real numbers will be used in the sequel. For the three-dimensional (3D) Euclidean space we use common $ {\mathbb R}^3 $. 
As is customary, let $ \partial \Omega $ denotes the closed boundary of a 3D domain (connected open subset) $ \Omega \subset {\mathbb{R}}^3 $ such that its closure is $ \overline{\Omega}= \Omega \cup \partial \Omega $. 

We shall study real-valued functions $f: \Omega \to \mathbb{R}_{+} $. Let $m \in \mathbb{N}\cup \lbrace 0 \rbrace $ (particularly $ C(\Omega):= C^0(\Omega)$) then $C^m (\Omega)$ stands for the class of $m$ times continuously differentiable functions on  $\Omega $ and $ C^m(\overline{\Omega}) $ denotes a space of functions having $ m $ continuous derivatives in $ \Omega \subset \mathbb{R}^3 $ and continuously extendable to the closure $ \overline{\Omega} $ of $ \Omega $. A function of class $ C^{\infty} (\Omega) $ is an infinitely differentiable functions. In the following, the term smooth function  on  $\Omega $ refers to a
$ C^{\infty} (\Omega) $-function.

Let $ \mathbf a =
({a}_{1},{a}_{2},{a}_{3}) \in \mathbb{R}^3 $ be a fixed translation vector (bold symbols are used for vectors or tensors in $ {\mathbb R}^3 $). For any point $ \mathbf x\in \mathbb{R}^3 $ define the translation operator (and its inverse one) $T_{\mathbf a} $ ($ T^{-1}_{\mathbf a}$)$: \mathbb{R}^3 \to \mathbb{R}^3$ by the relations
\begin{equation}
	T_{\mathbf a}\mathbf x = \mathbf x + \mathbf a \,, \quad T^{-1}_{\mathbf a}\mathbf x = \mathbf x - \mathbf a \,. \label{adg1c}
\end{equation}
It is well-known that translations form a subgroup of the 3D Euclidean motion group  $E(3)$ such that: corresponding composition operator  $T_{\mathbf c} = T_{\mathbf a} T_{\mathbf b}$ with translation vector $\mathbf c = \mathbf a + \mathbf b$; and identity operator $I:=T_{\mathbf 0} = T^{-1}_{\mathbf a} T_{\mathbf a}$~\cite{Miller77}.

\subsection{A model description}

To study diffusion-influenced reactions theoretically so-called {\it trapping model} is often used~\cite{Weiss94}. It has focused on studying diffusion of very small (point-like) reactants $B$ in heterogeneous media, in which the disperse phase behaves as a collection of immobile obstacles immersed in an inert continuous medium commonly known as {\it host medium}. This is the case, e.g., when obstacles diffuse much more slowly than reactants (they are of much larger size). Provided arbitrary but finite number of obstacles (sinks) are regularly or randomly distributed within a given region of surrounding medium, for this system of obstacles terms "ensemble" or "array" are usually used.

Standard treatments of the theory of diffusion-influenced reactions generally assume that the host medium is unbounded (open system~\cite{Beenakker85}), quiescent, isotropic and homogeneous, therefore mathematically it is well modelled by the 3D Euclidean space $ {\mathbb R}^3 $ (called {\it physical space} below). In turn reactants assumed to be identical non-interacting the point-like {\it Brownian particles} (from now on called $ B $-{\it particles} or $B$-{\it reactants}) diffuse among 3D non-overlapping immobile spherical obstacles embedded in some domains from the physical space.
In considering microscopic kinetics of diffusion-influenced reactions, it is necessary to distinguish between obstacles. Each obstacle is characterized by its geometrical and physicochemical parameters, which drastically influence on the reaction kinetics. 
Any particle of kind $A$ refers to the absorbing obstacle are usually called {\it trap} or {\it sink }~\cite{Weiss86}.
Further, for definiteness sake, we use the term "obstacles" to refer to the fully reflecting ones only, otherwise we shall use the most common term "sinks" solely. Despite its relative simplicity, the trapping problem is fundamental for the understanding of kinetics of diffusion-influenced reactions~\cite{Rice85}.

In this paper we deal with irreversible bulk contact diffusion-influenced reactions described by the simplest reaction
scheme~\cite{Rice85,Ovchinnikov89}
\begin{align}
	A+B\stackrel{k}{\longrightarrow }  A \cdot B \stackrel{k_{in}}{\longrightarrow} A + Products\, . \label{Zv00}
\end{align}
Here $ A\cdot B $ is so-called {\it encounter pair};  $k $ and $k_{in}$ stand for the {\it reaction rate constant} and  the {\it intrinsic reaction rate constant} for the reaction of the encounter pair $ A\cdot B $ to form inert products, respectively~\cite{Rice85}. Clearly, the rate $k_{in}$ describes surface reactivity of sinks $A$~\cite{Torquato02}, which we assume to be spherically symmetric.

Within the scope of scheme (\ref{Zv00}) the macroscopic action mass law equations for the bulk concentrations of reactants $ c_A\left(t \right)$ and $c_B \left(t \right)$ are
\begin{align}
	t_r {\dot c}_B\left( t \right) = - c_B \left( t \right) \,, \quad {\dot c}_A\left( t \right) =0 \,; \label{Int2}
\end{align}
where $ t_r := \left(k c_A \right)^{-1} $ is the characteristic time scale of the macroscopic kinetics. 

And, finally, it should be emphasized that one can assume infinite capacity approximation for sinks $A$ provided the initial concentration of $B$-particles appreciably exceeds that of sinks, i.e., when condition $ c_A\left(0 \right) \ll c_B \left(0 \right) $ holds~\cite{Ovchinnikov89}. To put it otherwise, reactions (\ref{Zv00}) are the catalytic type~\cite{Oshanin05}, when activity of the catalyst $ A $ remains unchanged. In this connection, reactions (\ref{Zv00}) are called sometimes catalytic degradation reactions.

\subsection{A microscopic Smoluchowski's theory}

The absorbing rate coefficient of $B$-particles on a given sink surface is of primary importance for the microscopic theory of  diffusion-influenced reactions (\ref{Zv00}). In his pioneering work of 1917 Smoluchowski, basing on Fick's laws of diffusion, proposed a rather simple calculation of the rate coefficient to describe absorption of the $B$-particles by the surface of a given fully absorbing spherical sink often referred to as a {\it test sink}~\cite{Brailsford76,Rice85}. 

Let us first treat the simplest case of the steady-state diffusion among randomly distributed static spherical sinks $A$ of the same radii $R$. So, for the {\it sink volume fraction} $\phi_A $ in dilute enough sink ensembles the following condition holds $ \phi_A:=  4 \pi c_A R^3/3 \ll 1$~\cite{Doktorov98}.
In other words diffusive interaction accepted to be weak and, therefore, it is commonly assumed that sinks mostly interact with the unbounded host medium~\cite{Voorhees01}.

Smoluchowski's theory corresponds to the limiting case of infinite dilution ($ \phi_A \to 0  $)~\cite{Venema89} and, therefore, only one-sink system should be treated. 

Let $ \Omega_1 \subset {\mathbb{R}}^3 $ be the test sink domain and consider its compliment
$\Omega_1^- := \mathbb{R}^3 \backslash \overline{\Omega}_1 $.  It is expedient to use the standard polar spherical coordinates  $ \lbrace O; r, \theta, \phi \rbrace $ ($ \theta $, $ \phi $ are the polar and azimuthal angles, respectively) coordinate systems attached to the origin $ O ({\mathbf 0}) $ coinciding with the test sink center. Consider  a current point $P \in \Omega_1^- $  with position vector $ \mathbf{r} $ such that  $r = \Vert {\mathbf{r}} \Vert  $ its radial distance, where $ \Vert \cdot \Vert $ stands for the common Euclidean norm. Plainly, for the spherical sink domain and relevant reaction surface  we have: $ \Omega_1:= \left\{ r < R \right\}$ and  $ \partial \Omega_1 = \left\{ r = R \right\}$. 
\begin{definition}\label{defin1}
	For structureless point-like $ B $ particles the {\it configuration space} is a 3D domain composed of all $ B $ particle positions $ \mathbf{r} $.
\end{definition}
Since diffusion of $B$ particles occur in the exterior of the sink, domain $\Omega_1^-$ is the configuration space 
and clearly $B$ particles local concentration is spherically-symmetric $ n_1 \left(r\right)$. Often function $ n_1 \left(r\right)$ is scaled with its bulk concentration $ c_B $ so that $ n_1 \left(r\right)/c_B $ treated~\cite{Rice85,McCammon13}.  
However, we are going to use another function namely {\it complementary normalized local concentration}, which seems to be more appropriate to describe subsequently diffusive interaction effects
\begin{align}
	u_1: \Omega_1^- \to \left(0, 1\right)\,, \quad \mbox{where} \quad	u_1\left( r\right) = 1- \frac{n_1\left( r \right)}{c_B}	\,.  \label{Zv02a}	
\end{align}
For brevity we will still call this function local concentration, if there is no confusion to be appeared.

in three dimensions there are 13 coordinate systems in which Laplace's equation is separable

According to Smoluchowski's theory, desired concentration $ u_1\left( r\right) $ is governed by the classical diffusion equation under appropriate boundary conditions. Thus, relevant exterior boundary value diffusion problem for a partially reflecting test sink reads~\cite{Rice85}
\begin{align}
	\nabla^2_r u_1 = 0\,  \quad \mbox{in} \quad \Omega_1^- \,,
	\label{Zv03a}\\
	\left.\left [ 4 \pi R^2 \, D\partial_{r}u_1  - {k}_{in}\left (1 - u_1 \right ) \right ] \right| _{r=R} = 0 \,, \label{Zv03b}\\
	\left. u_1 \right|_{r \rightarrow \infty } \rightarrow 0\, ,  \label{Zv03c}	
\end{align}
where $D$ is the translation diffusion coefficient and $\nabla_r^2 : = 2{r^{-1}}\partial_r  + {\partial_r}^2  $ is the radial part of Laplacian. For brevity, henceforth $ \partial_{\varsigma} $ stands for the partial derivative $ \partial/\partial {\varsigma} $ with respect to the independent variable $ \varsigma $. Note in passing that in physical literature condition (\ref{Zv03c}) is commonly referred to as the {\it outer boundary condition}~\cite{Karplus87,McCammon13}.

It has been well-known that solution to this problem is~\cite{Rice85}
\begin{align}
	u_1\left( r\right) = 1- \frac{R}{r}J_0^1\, .	  \label{Zv03d}	
\end{align}
In this formula we introdused dimensionless rate $J_0^1 := k_{CK}\slash k_S =\kappa_1\slash \left(1 + \kappa_1\right) < 1$,   
where $ \kappa_1 := k_{in}/k_S $. Hereinafter $ J_0^1 $ is the Collins-Kimball rate constant $ k_{CK}\left( k_{in} \right) $ normalized by the Smoluchowski rate constant  $ k_S $ given by
\begin{align}
	k_S = \lim_{k_{in} \to \infty}k_{CK} = 4 \pi R D = \Phi_1 \left( R \right)\slash c_B \,, \label{Zv01}
\end{align}
where  $ \Phi_1\left( R \right) $ is the total flux of $B$-particles on the test sink of radius $R$~\cite{Rice85}.

We face up to another important aspect. Lee and Karplus highlighted the fact that in terms of non-normalized concentration $ n_1 \left(r\right)$ the outer boundary condition is time dependent~\cite{Karplus87}, i.e.
\begin{align}
	\left. n_1 \right|_{r \rightarrow \infty } \rightarrow c_B \left( t \right)\, . \label{Zv01a}
\end{align} 
To clarify the matter we note that the steady state reaction regime occurs under the following condition $t_D \ll t \lesssim t_r $~\cite{Traytak13}, where $ t_D = R^2/3D $ being the {\it characteristic encounter time}~\cite{Doktorov19} and the frequency of reaction (\ref{Int2}) takes the form  $ t_r^{-1}= 4\pi DR J_1 c_A  $. 

Strictly speaking, one should consider the above boundary value problem (\ref{Zv03a})-(\ref{Zv03c}) as a quasi-steady state one~\cite{Beenakker85}. That is consistent with the fact that the characteristic time scale of the change in $ c_B \left( t \right) $ (see Eq. (\ref{Int2})) is $ t_r \gg t_D = t_r J_1 \phi_A $. 	
Really, one can see that "fast" (microscopic) time evolution relaxation towards a quasi-steady-state occurs for times $ t = {\cal{O}}(t_D)$ while changing of bulk concentration $ c_B \left( t \right) $ takes place at "slow" (macroscopic) times  $ t = {\cal{O}}(t_r)$. 

Concluding the section note that we can treat recalled here one-sink case as a reference problem for further generalization to the case of $N$ sinks when diffusive interaction cannot be neglected.

\section{Mathematical methods to study diffusive interaction}\label{sec:DiffInter}

Subsequently, one-sink diffusion theory was generalized in many aspects. This section focuses on previous theoretical works developing above theory to account for the diffusive interaction between sinks. 

Smoluchowski's theory, in view of Eq. (\ref{Zv01}), leads to the total flux of $B$'s into the whole ensemble of $N$ identical uncharged sinks as the additive sum $\Phi_N =  N \Phi_1  \left( R \right)$~\cite{Richards86}. Nevertheless, over the years it has been realized that this theory faces challenges under attempts to understand various reaction-diffusion phenomena in systems comprising many sinks with small sink-sink separations. Indeed, each sink, due to chemical reaction (\ref{Zv00}) occurring upon its surface, exerts influence on the $ B $'s concentration field around other  sinks. Simply speaking an influence arises between the sinks as a result of the fact that any sink "feels" the self-consistent diffusion field of $B$-particles, determined by the entire array of sinks. Obviously this fact is caused by the law of conservation of  $B$-particles in a host medium.

It must be recognized that, up to now, there has been no established terminology concerning above effects. The reason is that since these effects have wide applications in various fields, various technical languages, including definitions and notations were used. Really, note that a variety of terminology is used in the literature even to define such kind interaction. Besides the following terms are used: (1) competition~\cite{Frisch52,Ham58,Deutch76,Strieder04,Green06,McCammon13,Biello15}; (2)  diffusive (diffusional) interaction~\cite{Carstens70,Beenakker85,Dubinko89,Traytak91,Traytak92,Marsh95,Voorhees01,Piazza19,Grossier22}; (3) effect of multiparticle interaction~\cite{Labowsky78,Marberry84,Zinchenko94}; (4) Laplace's interaction~\cite{Mo96}; (5) concentration effects~\cite{Lee14}; (6) collective effects~\cite{Carrier16,Stark18, Masoud21};	(7) reactive interference~\cite{Lee20}; (8) shielding effects~\cite{Schofield20}; (9) cooperative effects~\cite{Wang22} and even (10) chemical interaction~\cite{Stark18}.
That is why throughout the paper we shall clarify some of relevant terms which are used in the literature.

Regarding above-mentioned note that the use of the term "collective effects" seems an unfortunate one because of its frequent utilization in particle accelerators physics. At the same time we believe that the among pointed terms, term "Laplace interaction" (or, generally speaking,  "Laplacian transport"~\cite{Zagrebnov11}) best describes the essence of the interaction under consideration for the steady-state diffusion, heat transfer and Stokes hydrodynamics~\cite{Kim91}. Nonetheless, for definiteness sake we, as before, use only term {\it diffusive interaction}, taking in mind the fact, that obtained results may be applied to the another counterparts~\cite{Piazza19}.

As far as we know  Frisch and Collins pioneered in studying "competition among sinks for the diffusing molecules" to the diffusion-influence processes, particularly growth of aerosol particles by condensation~\cite{Frisch52}. Among early approaches it is also worth mentioning Ham's theory based on the Wigner-Seitz cell model~\cite{Ham58}. Applying a variational procedure Reck and Prager had first established rigorous upper and lower bounds on the rate of a diffusion-controlled reactions~\cite{Reck65}. Detailed exposition of this approach with applications can be found in the book Ref.~\cite{Torquato02}.

Later, during many years, a serious effort has been mounted toward calculation of the microscopic reaction rates, taking into account multisink effects. The monopole and dipole approximations for the steady state diffusion and reactions of $B$-particles in dense ensembles of fully absorbing spherical sinks has been extensively investigated by many authors~\cite{Borzilov71,Weins73,Brailsford76,Deutch76,Felderhof76,Marqusee84,Rice85,Felderhof86,Traytak89,Marsh95,Strobe96,Mandyam98,Tsao01,Biello15,Stark18}. However, it was noted long ago: "... the interactions between sinks for the competitive consumption of the solutes are not exactly accounted for in these approximations"~\cite{Zheng89}. Using the moment scattering, Bonnecaze and Brady derived analytically the reaction rate for cubic arrays of sinks up to the quadrupole level~\cite{Brady91}. Incidentally, they have emphasized there: "To accurately compute the effective reaction rate at high volume fractions, higher order many-body multipole interactions are required."   

Most attention has been concentrated on an important particular case of two sinks. While the above illustration relied on an example with two spherical sinks that could alternatively be solved analytically by using bispherical coordinates~\cite{Rice85,Green06} (see also Sec.~\ref{sec:examples}).  

For $ N>2 $ the problem was solved by the images and reflections methods analytically~\cite{Deutch76,Labowsky78}. It is noteworthy that the connection between the method of images and the method of reflections was discovered in Ref.~\cite{Traytak03}. An important point is that within the framework of the classical  method of images, the corresponding compensating solutions are based on point charge potentials inside sinks and, thereby, allow to obviate the need for the use of addition theorems for solid harmonics~\cite{Traytak03}.

Concerning theoretical methods to attack the multi-particle diffusion problem Ratke and Voorhees (Ref.~\cite{Voorhees02}) noted: "We shall develop a solution to the diffusion problem that is consistent with the particles having a spherical morphology. This can be done in a self-consistent manner using multiple scattering theory~\cite{Marqusee84}, irreducible Cartesian tensor~\cite{Beenakker86} or using boundary integrals and multipole expansions~\cite{Voorhees94}.". On this head we note that multiple scattering theory~\cite{Marqusee84} based on the microscopic monopole approximation and  multipole expansions method was performed in spherical coordinates~\cite{Voorhees94}. In its turn the irreducible Cartesian tensor approach~\cite{Beenakker86} is nothing more but a scalar version of well-known induced force method (see Sec.~\ref{sec:discussion}).

Later on influence of many neighboring sinks on diffusion-controlled reactions was theoretically investigated in Ref.~\cite{McCammon13}. Particular focus has been given there to the fact that: "... in most analytical works, only two-sphere cases were mainly considered in the calculations of physical quantities such as rate constant. This is because of difficulty in solving differential equations with more than three spheres (or with many local boundary conditions)." On this basis authors, "to avoid the mathematical difficulty encountered in analytical approach", have been solved the corresponding boundary value problem for the 3D diffusion equation by means of the finite element method~\cite{McCammon13}. However, it turns out that the diffusive interaction is harder to describe by this method especially in the case of fully reflecting obstacles. 

We emphasize that interest to the study of the diffusive interaction was rekindled last decades due to problems on crowded systems~\cite{Vazquez10,Lee20} and active transport. More recently, there has been a growing interest to theoretical and experimental investigations on the diffusive interaction effects of neighboring droplets during their evaporation on a surface (substrate). Particularly in recent  experiments, there was evidence that  "...neglecting the diffusive interactions can lead to severe inaccuracies in the measurement of droplet concentration..."~\cite{Grossier22}. This circumstance becomes really important bearing in mind the fact that "droplet evaporation on surfaces is ubiquitous in nature and plays a key role in a wide range of industrial and scientific applications..."~\cite{Grossier22}. Note in passing that quasi-steady-state diffusion to an assembly of slowly growing truncated sphere on a substrate may be also treated with the help of spherical multipoles~\cite{Bobbert87}.

At the same time we especially stress that the diffusive interaction effects  under consideration one should not be confused with so-called  {\it diffusion interaction parameter} widely studied experimentally in recent years (see Ref.~\cite{Betzel22} and the references therein). The latter effects are the connected with influence of concentration effects to the diffusion coefficient, which significantly is manifested in highly concentrated protein solutions.

Early works on the diffusive interaction were comprehensively reviewed in Ref.~\cite{Calef83}, whereas the current status of research on this subject one can find in the recent survey~\cite{Piazza19}.

Thus, accurate theoretical description of the problems involving diffusive interaction in multi-particle ensembles of  spherical obstacles and sinks is a long-standing challenge due to their many-body nature.

\section{Statement of the general problem}\label{sec:state}

As we pointed out above the diffusive interaction is the most manifested in the steady state regime. Accordingly, when studying the diffusive interaction this circumstance allows us to simplify significantly the mathematical problem, ignoring the time-dependent effects. 

A rigorous formulation of problems on the microscopic theory of diffusion-influenced reactions comprises: (a) specifying the geometry of the configuration space; (b) account of the reactants properties and the host medium involved; (c) using adequate diffusion system of the continuity equations and constitutive relation; (d) prescribing appropriate boundary conditions. It is important to note from the outset that by solution to the diffusion-reaction problem, we mean solely its classical solution.

Basicly our mathematical statement of the problem is similar to that given in Ref.~\cite{McCammon13} (attention is drawn to the fact that sinks and diffusing particles denoted there as $ B $ and $A$, respectively). 

The diffusive interaction substantially depends on a given configuration of reactive and inert boundaries~\cite{Piazza19}, so first of all we specify the geometrical part of the problem.

\subsection{Domains definition}

Consider a collection of arbitrary but finitely many $ N \geq 1 $ microscopic spherical sinks of different radii $R_i$, immersed within an unbounded host fluid medium ${\mathbb{R}}^3 $. Hereafter the sinks are labelled sequentially by the index $ i = \overline{1, N} $, where symbol $\overline{l, m}$ means that all integers (including $\lbrace \infty \rbrace $) from $ l $ to $ m $ are taken their values successively, i.e.
$\overline{l, m}:=l, l+1, \ldots, m$, where $ l (< m \leq \infty ), m \in \mathbb{Z} $.

Geometrically we treat the $ i $-th sink as a 3D body occupying a spherical domain (open ball) $ {\Omega}_i \subset {\mathbb{R}}^3 $.  Denote by $ \overline{\Omega}_i $ the corresponding closed bounded spherical domain, supposing that all $\overline{\Omega }_i$ are non-intersecting or non-touching, i.e. $	
\overline{\Omega }_i\cap \overline{\Omega }_j= \varnothing $  for $ i(\neq j)=\overline{1,N} $. So the total domain occupying by array of $N$ sinks is $\Omega^+ := \bigcup_{i=1}^N {\Omega}_i$.
Clearly this domain has multiply connected boundary $\partial \Omega^+ = \bigcup\limits_{i=1}^N {\partial\Omega }_i$, where corresponding boundary $\partial \Omega _i$ ($i$-th sink reaction surface~\cite{Rice85,Ovchinnikov89}) is the $ i $-th {\it connected component} of the total boundary $\partial \Omega^+ $.

On the other hand, it is apparent that the diffusion problem should be posed on the whole 3D physical space $ \mathbb{R}^3 $ excepting the total domain occupying by sinks, i.e., in sink-free region~\cite{Torquato02}. Hence the required domain outside all sinks is the {\it exterior domain} with respect to sinks: $ \Omega^- := \mathbb{R}^3 \backslash \overline{\Omega}^+ \subset \mathbb{R}^3 $ and the whole physical space $ \mathbb{R}^3 $ can be naturally partitioned into two complementary subdomains  $ \Omega ^+ $ and $ \Omega ^- $.

Plainly, $ \Omega^- $ is the set of accessible configurations for $B$'s so, according to above definition~\ref{defin1}, we term this domain as a {\it configuration space}. It is evident that diffusive interaction is highly affected by geometry of the configuration space $ \Omega^- $.
\begin{example}
	For instance consider two sinks labeled  $i$ and $ j $ with radii $ R_i $ and $ R_j $, as depicted in Fig.~\ref{fig1}. Here we placed a particle $B$ is into some current point $ P $ of the configuration space  $\Omega^- := \mathbb{R}^3 \backslash \left( \overline{\Omega}_i \cup\overline{\Omega}_j \right) $.
\end{example}
\begin{figure}[h]
	\centering
	\includegraphics[width=0.67\textwidth]{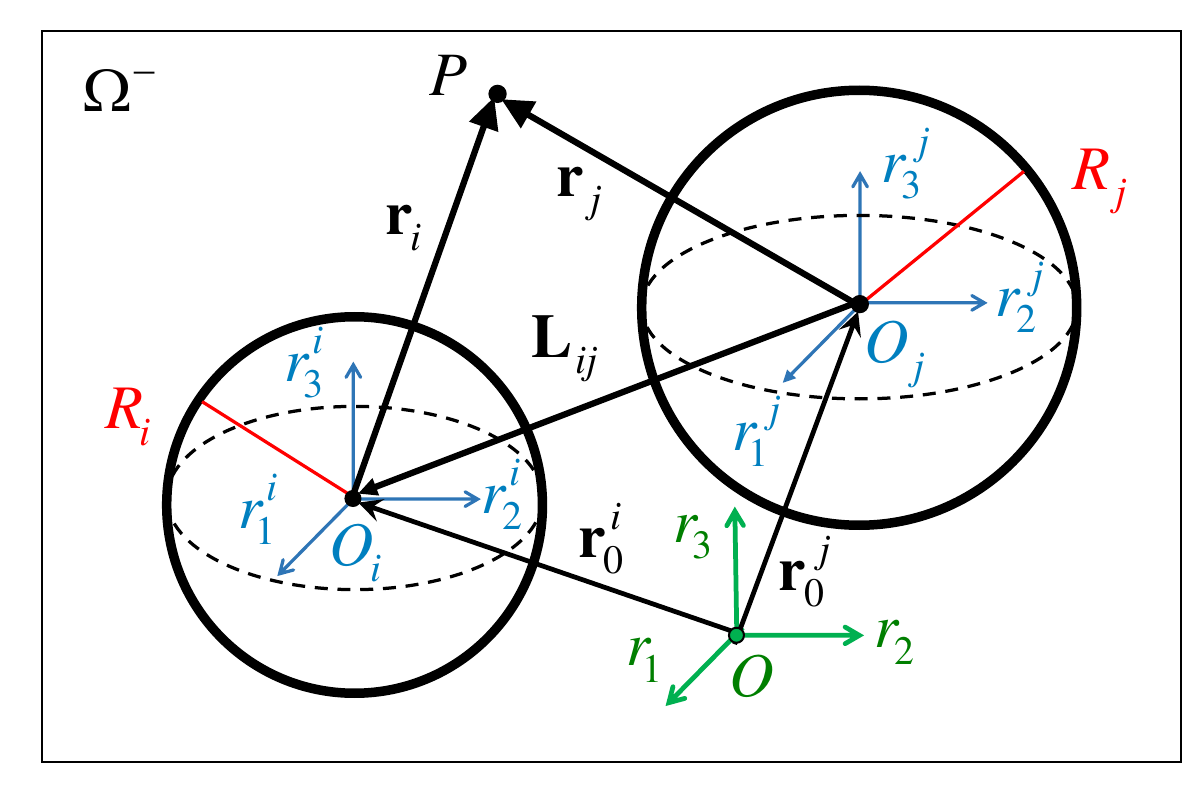}
	\caption{Definition sketch for two-sink array. The configuration space $ \Omega^- \subset \mathbb{R}^{3} $ corresponds to the given microstructure $X^{(2)}$ with a $B$ particle located at point $P \in \Omega^-$. Cartesian coordinate systems: global $ \lbrace O; {\mathbf{r}}\rbrace $ (green); local in $i$-th sink $ \lbrace O_{i}; {{\mathbf r}_{i}} \rbrace $ and in $j$-th sink $ \lbrace O_{j}; {\mathbf r}_{j} \rbrace $ (blue); sink radii $ R_{i} $ and $ R_{j} $ (red)}\label{fig1}
\end{figure}
Let $ \mathbf{r}:= \overrightarrow{OP} $ be a vector that represents the position of a current point $ P $ in $ \Omega^- $ with respect to some fixed point $ O $. Introduce a global Cartesian coordinate system, taking point $ O $ as the origin: $ \lbrace O; \mathbf{r} \rbrace $ and attach the standard orthonormal Cartesian basis $  \lbrace {\mathbf e}_{\alpha} \rbrace_{\alpha = 1}^{3} $ to this origin. The relaited components of the position vector $ {\mathbf r} $ we denote by $ r_\alpha $ ($\alpha=1, 2, 3 $) or in ordered triple form by $ \left(r_1, r_2, r_3 \right) $. 

Using the global coordinates, we designate by $\mathbf{r}^i_0 $ the position of the center of $ i $-th sink and by $ \mathbf{L}_{ij} := \overrightarrow{O_j O_i} = \mathbf{r}^i_0  - \mathbf{r}^j_0  $ the separation vector between the two sinks centers, pointing from $ j $-th sink toward $ i $-th sink (see Fig.~\ref{fig1}).

Clearly, any $ i $-th sink spherical domain is:
$\Omega_i= \lbrace  \mathbf{r} \in {\mathbb{R}}^3: \Vert \mathbf{r} - \mathbf{r}^i_0 \Vert < R_i \rbrace$ with the boundary $\partial \Omega _i = \lbrace  \mathbf{r} \in {\mathbb{R}}^3: \Vert \mathbf{r} - \mathbf{r}^i_0 \Vert = R_i \rbrace$. 

We shall consider diffusion of $B$'s in the configuration manifold $ \Omega ^{-} $ formed by an arbitrary but finite ensemble of $ N $ spherical sinks with fixed radii $ \left\{R_i\right\}_{i=1}^N $ centered at positions $ \left\{{\mathbf r}^i_0 \right\}_{i=1}^N $ immersed in the 3D host medium. Hence the geometry of this manifold is completely determined by $ N $-{\it sink configuration}~\cite{Torquato02} or  briefly {\it microstructure}~\cite{Brown16}:
\begin{align}
	X^{(N)}:=\left\{ \left( {\mathbf r}^i_0,  R_i \right) \right\}_{i=1}^N \,. \label{Mist1}
\end{align} 

In Ref.~\cite{Traytak03} we have shown that the 3D smooth manifold is a relevant mathematical concept to describe diffusion-influenced reactions in domains with multiply connected boundaries. 
Bearing in mind further development of the GMSV to sinks of different types of canonical domains (see below definition~\ref{definCan}) we consider here a general formalism based on the manifold conception~\cite{Traytak03}. We shall show here that the manifold conception arises simply for the above microstructure (\ref{Mist1}).
Nevertheless, in our view, the much more surprising thing is that researchers in various applied sciences still have no use in boundary value problems similar to treated here this natural and convenient language.  Especially since, paraphrasing Sawyer~\cite{Sawyer71}, one may say, "We have all worked with manifolds but, like Moli\`ere's character who had been speaking prose all his life, we have not been aware of it."

\subsection{Manifold structure on $\Omega^{-}$}\label{subsec:addtheorDef}

Let us introduce the manifold structure on the configuration space $\Omega ^{-} $. First let us introduce open complements of the sink domains $ \Omega_i $ in the physical space  $ {\mathbb R}^3 $, which play an important role in our study: $\Omega ^{-}_i={\mathbb R}^3\backslash \overline{\Omega}_i$. It is clear that for all arrays ($ N \geq 1 $) we have: $ \Omega^- = \bigcap\limits_{i=1}^N \Omega_i^- \subseteq \Omega_i^- $. 

Configuration space $\Omega ^{-} $ is a 3D {\it smooth manifold} (strictly speaking sub-manifold in $\mathbb{R}^3$)~\cite{Tu11} since:
(a) there is its coverage $\bigcup\limits_{i=1}^N \Omega_i^- \supset \Omega ^{-} $; (b) $\Omega ^{-}$ is a metric space as a subset of $ ({\mathbb R}^3, \Vert \cdot \Vert) $; (c) there exists smooth homomorphisms  $  \varphi_i: \Omega_i^- \to {\mathbb R}^3 $; (d) For intersection $ \Omega ^{-}_{ij}:= \Omega ^{-}_i \cap \Omega ^{-}_j \neq \varnothing $ the overlap map:  $ \varphi_{ij}:=\varphi_j \circ {\varphi_i}^{-1}: \varphi_{i}\left(\Omega ^{-}_{ij}\right) \to \varphi_{j}\left(\Omega ^{-}_{ij}\right)$ is smooth for each pair of indices $ i(\neq j)=\overline{1,N}$.

A {\it coordinate chart} (or just a {\it chart}) on $\Omega ^{-} $ is a pair $\left( \Omega ^{-}_i,\varphi _i\right) $.  A collection of all charts $\left( \Omega ^{-}_i,\varphi _i\right) $ ($ i=\overline{1,N}$) is called a smooth atlas on $\Omega ^{-} $, which defines the manifold structure on $\Omega^{-}$. For short in the following  we shall simply term $ \Omega^{-}_i $ the {\it $i$-th chart}. It is clear that boundary $\partial \Omega ^{-} $ is a 2D manifold with empty boundary.

Consider the position vector $ {\mathbf r} $ of point $P \in \Omega^{-}_i $, then $ \varphi_i \left( \mathbf{r} \right) = ( r_1^i, r_2^i, r_3^i) \in \mathbb{R}^3 $. These $3$ real numbers are called {\it local coordinates of point $P$ in the chart $ \Omega^{-}_i $}. The corresponding local Cartesian coordinate system for any $i$-th sink it is convenient to designate by $ \lbrace O_i; r^i_1, r^i_2, r^i_3 \rbrace \equiv \{ O_i;{\mathbf r}_i\} $ with the same orientation as the global one and the origin $ O_i ({\mathbf r}^i_0) $ at the center of the $i$-th sink (see Fig.~\ref{fig1}). So, for any fixed $ i $ we can choose standard local orthonormal Cartesian basis $  \lbrace {\mathbf e}_{\alpha}^i \rbrace_{\alpha = 1}^{3} $ attached to the origin $ O_i ({\mathbf r}^i_0) $, such that $ r^i_{\alpha} $ are coordinates in this basis. Introdused local coordinates are useful for computation of the diffusion field in arrays wth microstructure (\ref{Mist1}). 

Let us take above position vector $ {\mathbf r} $ for the point $ P $ and consider it with respect to two local Cartesian coordinate systems $ \lbrace O_{i}; {{\mathbf r}_{i}} \rbrace $  and $ \lbrace O_{j}; {\mathbf r}_{j} \rbrace $. Corresponding point $ P $ position vectors are 
\begin{align}
	\varphi_i \left( \mathbf{r} \right) = {\mathbf{r}}_i = \mathbf{r} - \mathbf{r}^i_0 \quad \mbox{and} \quad  \varphi_j \left( \mathbf{r} \right) = \mathbf{r}_j=\mathbf{r}-\mathbf{r}^j_0 \,. \label{it2a}
\end{align}
It is also evident that $ \varphi _i^{-1}\left( \mathbf{r}_i\right) = \mathbf{r}_i + \mathbf{r}_0^i = \mathbf{r} $.
Perform the transition between these local coordinates given by the overlap map $ \varphi_{ij}=\varphi_j \circ {\varphi_i}^{-1}:{\mathbf{r}}_i  \rightarrow {\mathbf{r}}_j $ with the rule (see Fig.~\ref{fig1}):
\begin{align}
	{\mathbf{r}}_j= \varphi_{ij}\left({\mathbf{r}}_i\right) = \varphi_j\left(\left({{\varphi_i}^{-1} \left({\mathbf{r}}_i\right)}\right)\right) = \mathbf{r}%
	_i+\mathbf{L}_{ij}
	\,, \quad \mbox{where} \quad \mathbf{L}_{ij}=\mathbf{r}^i_0-\mathbf{r}^j_0\,. \label{it2}
\end{align}

Hereafter $ \mathbf{L}_{ij} $ being the vector, connecting centers of sinks $ i $ and $ j $. In this work we also use the corresponding unit vector $ {\widehat{\mathbf{L}}}_{ij} := {\mathbf{L}}_{ij}/  L_{ij} $, where $ L_{ij} $ is the distance between above centers $ L_{ij}= \Vert {\mathbf{L}}_{ij} \Vert = \Vert \mathbf{r}^i_0-\mathbf{r}^j_0 \Vert = \left\| {\mathbf r}_i-{\mathbf r}_j\right\|$. 
Transformations (\ref{it2a}) and (\ref{it2}) are clearly illustrated by the following triangle commutative diagram: 
\begin{center}
	\begin{tikzpicture}
		\matrix (m) [matrix of math nodes,row sep=3em,column sep=4em,minimum width=2em]
		{
			\lbrace O_i;{\mathbf{r}}_i \rbrace & \lbrace O_j;{\mathbf{r}}_j \rbrace \\
			\lbrace O;{\mathbf{r}} \rbrace& \\};
		\path[-stealth]
		(m-1-1) edge node [left] {${\varphi_i}^{-1}$} (m-2-1)
		edge [double] node [below] {$\varphi_{ij}$} (m-1-2)
		(m-2-1) edge node [right] {$\quad\varphi_j$} (m-1-2);
		
	\end{tikzpicture}
\end{center} 
\begin{remark}
	Note that simple linear connection between local coordinates (\ref{it2}) arose due to the use of Cartesian coordinates. Really, the application of polar spherical coordinate systems leads to the corresponding nonlinear dependence.  	 
\end{remark}
In view of the foregoing, the configuration space $\Omega ^{-} $ can be called the {\it configuration manifold}.

An additional point to emphasize is that theoretical calculations are often used the distance between surfaces of domains $ \Omega_i $ and  $ \Omega_j $, i.e. $h_{ij} := L_{ij} - \left( R_i + R_j\right)$ (see Fig.~\ref{fig1}). Hence the condition describing configurations of $ N $ sinks, which are not intersecting or touching may be written explicitly as  
\begin{align}
	h_{ij} > 0 \quad \mbox{for all} \quad i(\neq j)=\overline{1,N} \,. \label{Zv00b}
\end{align}

\subsection{Governing equation}

To describe the steady state microscopic diffusion-reaction model introduce a local concentration (number density) of $B$-reactants, $ n:\Omega^- \rightarrow \mathbb{R}_+ $ depends on a given microstructure $X^{(N)}$ (\ref{Mist1}) and also different reactivities of sinks. To describe the latter it is convenient to introduce the {\it set of sinks reactivities} $k_{in}^{(N)}:=\left\{ {k}^i_{in} \right\}_{i=1}^N $, where rate ${k}^i_{in}$ corresponds to the $i$-th sink reactivity. Thus, in a fixed global Cartesian coordinate system  $ \lbrace O; \mathbf{r} \rbrace $, extending notations of Refs.~\cite{Brailsford76,Torquato02}, one can write this function as $ n \left( P; X^{(N)},\,  k_{in}^{(N)} \right) $. However, for short we denote it just by  $n_N \left(P\right)$. Moreover, to investigate the effects of diffusive interaction as in the case of one sink Eq. (\ref{Zv02a}) we introduce  normalized concentration of $ B $ particles related to their rescaled local concentration 
\begin{align}
	u_N: \Omega^- \to \left(0, 1\right)\,, \quad \mbox{where} \quad u_N\left( P\right) = 1- \frac{n_N\left( P\right)}{c_B}\,.	  \label{int0}	
\end{align}
Assuming that rescaled Fick's local diffusion flux constitutive relation holds~\cite{Rice85} 
\begin{align}
	\mathbf j = D {\bm\nabla} u_N\left( P\right) \quad \mbox{in} \quad \Omega ^{-}  \label{tds1a0a}
\end{align}
the rescaled local concentration $ u_N\left( P \right) $ is governed by the $B$-particles
conservation law which leads to the Laplace equation
\begin{equation}
	{\bm \nabla} ^2 u_N\left( P\right) = 0 \quad \mbox{in} \quad \Omega ^{-} \, . \label{tds1a1}
\end{equation}
Hereinafter as usual $ {\bm \nabla} := \sum_{\alpha=1}^{3} {\mathbf e}_{\alpha} {\partial_{\alpha}} $ and  $ {\bm \nabla}^2 $ stand for the  gradient operator and the Laplacian on $ \mathbb{R}^3 $, respectively. 

To describe diffusion-reaction phenomena, Eq. (\ref{tds1a1}) should be subjected to the appropriate boundary conditions on the sink surfaces $ \partial \Omega_i $ and condition at infinity.

\subsection{General boundary conditions}

The present work is concerned with the whole possible range of the intrinsic rate constants $ 0 \leq {k}^i_{in} < + \infty $ and in general case we can impose the local form of inhomogeneous Robin boundary conditions 
\begin{equation}
	\left.\left [ 4 \pi R_i^2 D \, \hat{\mathbf n}_i \cdot {\bm \nabla}u_N - {k}^i_{in} \left (f_i\left( P_i\right)-  u_N  \right)\right]  \right| _{\partial \Omega_i} = g_i\left( P_i\right) \,,  \quad i=\overline{1,N}\,; \label{tds1a}
\end{equation}
where $ \hat{\mathbf n}_i $  being the outer-pointing unit normal vector with respect to $ \partial \Omega_i $ (see Fig.~\ref{fig2}) and point $P_i \in \Omega_i$. We assume for simplicity that $f_i\,, g_i \in C(\partial \Omega_i)$. 
\begin{figure}[h]
	\centering
	\includegraphics[width=0.63\textwidth]{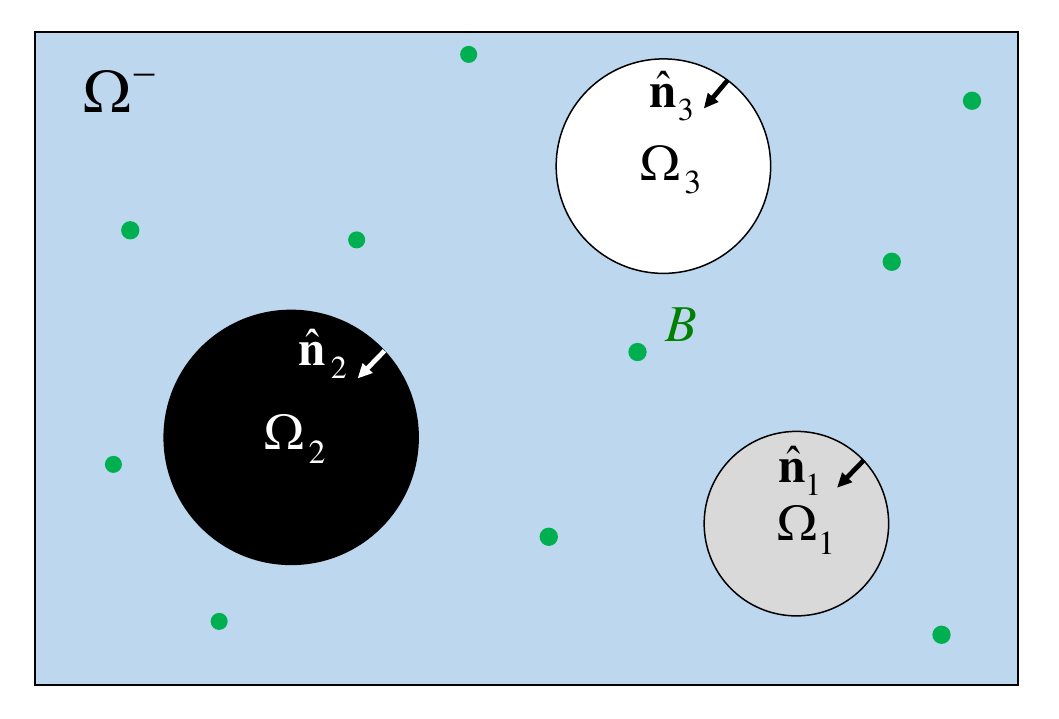}
	\caption{A geometrical sketch for three kinds of sinks in the configuration manifold $ \Omega^- $: partialy reflecting $ \Omega_1 $ (gray), fully absorbing $ \Omega_2 $ (black) and fully reflecting $ \Omega_3 $ (white); $ \hat{\mathbf n}_i$ are the corresponding outward to $ \Omega^- $ unit normal vectors; and small green circles depict $B$ particles}\label{fig2}
\end{figure}

For applications it is expedient to distinguish two particular cases of the general boundary conditions (\ref{tds1a}).

(i) In the case when $ k_{in}^i \to 0 $ conditions (\ref{tds1a}) give inhomogeneous Neumann boundary conditions 
\begin{align}
	4 \pi R_i^2 D\left. \hat{\mathbf n}_i \cdot {{\bm \nabla}u_N}  \right| _{\partial \Omega_i} = g_i\left( P_i\right) \,.  \label{tds5a}
\end{align}

(ii) Provided $k_{in}^i \rightarrow \infty$, conditions (\ref{tds1a}) are reduced to inhomogeneous Dirichlet boundary conditions
\begin{align}
	\left.  u_N  \right| _{\partial \Omega_i}=f_i\left( P_i\right). \label{tds2a}
\end{align}

For uniqueness of a solution to Laplace's equation (\ref{tds1a1}) under conditions (\ref{tds1a}) one should require satisfying the regularity conditions at infinity: $u_N\left( P\right) = {\cal{O}}(\Vert {\mathbf r}\Vert^{-1})$ and ${{\bm \nabla}u_N}\left( P\right) = {\cal{O}}(\Vert {\mathbf r}\Vert^{-2})$ as $\Vert {\mathbf r}\Vert \rightarrow \infty $~\cite{Tikhonov63,Krutitskii99}. These conditions are equivalent to the following {\it regularity condition at infinity}
\begin{align}
	\left.  u_N  \left( P\right) \right|_{\Vert {\mathbf r}\Vert \rightarrow \infty } \rightrightarrows 0\,   \label{rev5i}
\end{align}
and it uniquely determines the function $ u_N  \left( P\right) $, which is harmonic in the infinite domain $ \Omega ^{-} $~\cite{Tikhonov63}. Hereafter the sign $\rightrightarrows $ denotes the uniform limit. 

Thus, the exterior inhomogeneous Robin boundary value problem (\ref{tds1a1}), (\ref{tds1a}) and (\ref{rev5i}) is formulated. The existence and uniqueness of the classical solution to the posed problem $ u_N  \left( P\right) $ are a matter of common knowledge~\cite{Tikhonov63}. Furthermore, it is well-known that the smoothness class of solutions $ u_N  \left( P\right) $ depends on a boundary operator: (a) for the Robin and Neumann problems $ u_N \left({P}\right)\in C^2\left(\Omega^-\right) \cap C^1\left(\overline{\Omega}^-\right)$; (b) for the Dirichlet  problem  $ u_N \left({P}\right)\in C^2\left(\Omega^-\right) \cap C\left(\overline{\Omega}^-\right)$~\cite{Ladyzhenskaya68,Krutitskii99}.

\subsection{Boundary conditions for diffusion-influenced reactions}

It must be emphasized that the general boundary value problem (\ref{tds1a1}), (\ref{tds1a}) and (\ref{rev5i}) describes a vast amount of diffusion-reaction phenomena, occurring in nature and industry (see Sec.~\ref{sec:motivation}).
Although the approach under consideration may be applied to that problem, for the present paper we confine ourselves to study an important special case: $f_i\left( P_i\right) \equiv 1$ and $g_i\left( P_i\right) \equiv 0 $, describing the microscopic Collins-Kimball diffusion-influenced reactions~\cite{Rice85}.

Specific values of ${k}^i_{in}$ should be known experimentally or by means of more detailed kinetic theory, e.g., based on the Fokker-Planck-Klein-Kramers equation~\cite{Rice85}. Provided magnitudes of ${k}^i_{in} $ are positive and finite, we deal with {\it diffusion-influenced reactions} on so-called {\it partially reflecting (absorbing) sinks}~\cite{Kapral77,Kapral79}.  In the case of very large values of ${k}^i_{in}$ (i.e., as ${k}^i_{in} \to + \infty $), the sinks are called {\it fully (perfectly) absorbing}, and the corresponding reactions, occurring on them, are known as {\it fully diffusion-controlled}~\cite{Rice85}. Physicall that means that the activation barrier appeared to be very small~\cite{Vauthey17}.  Particular attention should be given to the case of {\it fully reflecting obstacles} (when ${k}^i_{in}=0$), which becomes also important for dense arrays of sinks and obstacles. Formally speaking, we can treat a fully reflecting obstacle as a sink with the zero intrinsic rate. 

Of special interest in many applications is the case when sinks possess different surface reactivity for a  microstructure $X^{(N)}$. To put it another way when intrinsic rate constants $ k_{in}^i $ are different for at least some numbers from the set $\left\{ 1 < i \le N  \right\} $~\cite{Doktorov19}. At first glance it would seem that this case is described by boundary conditions (\ref{tds1a}) automatically. However, solution of the boundary value problem (\ref{tds1a1}), (\ref{tds1a}) and (\ref{rev5i}) may be faced subtle mathematical difficulties~\cite{Krutitskii99,Traytak07}.

In this connection we believe that the mixed boundary value problems should be distinguished by two types of boundary conditions. We adopt the following definitions.
\begin{definition}\label{defin2}
	If the boundary conditions differ on different parts of a given connected
	component of a boundary $ \partial \Omega_i$ these conditions are referred to as {\it proper mixed boundary conditions}.
\end{definition}
Physically the proper mixed boundary condition means that a given connected component $ \partial \Omega _i $ is characterized by a prescribed heterogeneous reactivity. So from analytical viewpoint one should define a piecewise constant function ${k}^i_{in}: \partial \Omega_i \rightarrow  [0  + \infty) $. 
\begin{definition}\label{defin3}
	Provided different boundary conditions are posed on different connected components they are termed {\it improper mixed boundary conditions}.
\end{definition}
Therefore, for the improper mixed boundary conditions we have: $ \left.  {k}^i_{in} \right|_{\partial \Omega _i} = const $.
Figure~\ref{fig2} shows a geometric sketch of a simple array consisting of three chemically distinct reaction surfaces described by the improper mixed boundary conditions.

It is significant that, generally speaking, proper mixed boundary value problems are seldom amenable to analytical treatment because of their complexity, therefore, henceforth we shall treat here only improper mixed problems, omitting word "improper". For more details regarding derivations of improper mixed boundary conditions, the reader is referred to Refs.~\cite{Traytak07,Schuss15}.

In concluding this section, we note, that crowding and geometrical constraints effects on diffusion-influenced reactions~\cite{Lee20} may be treated if one considers microstructures $ X^{(N)} $ containing $ 1 \leq N_0 \leq N $ inert obstacles.
The presence of fully reflecting obstacles might strongly influence the reaction rate when the number of these obstacles is large enough~\cite{McCammon13,Lee20} (see also numerical calculations of the reaction rates for the mixed Dirichlet-Robin problem at $N=3$ in Sec.~\ref{sec:examples}).

\subsection{The dimensionless formulation of the problem}

Before subsequent mathematical study it is expedient to reduce the original boundary value problem to a non-dimensional form. 
In its turn, for spatial dimensionless independent variables we shall utilize corresponding local Cartesian coordinates $\left( O_i; \mathbf{r}_i \right) $ at the $ i$-th sink of radius $ R_i $ as a characteristic unit:
\begin{align}
	{\bm \xi}_i=\mathbf{r}_i/R_i\,, \quad  {\bm \xi}^i_0=\mathbf{r}^i_0/R_i \quad \xi^i_{\nu}=r^i_{\nu}/R_i\,, \quad \xi_i=\Vert {\bm \xi}_i \Vert \,.
	\label{int1}	
\end{align} 
Besides these dimensionless local coordinates for any pair of sinks $ i $ and $ j (\neq i) $ the following important parameters~\cite{Traytak06} 
\begin{align}
	\varepsilon_{ji}=R_{i}/L_{ij} < 1\,, \quad \varepsilon_{ij}=R_j/L_{ij} < 1  
	\label{int2e}	
\end{align}
are naturally arisen. Clearly, parameters (\ref{int2e}) are totally determined by a given microstructure $ X^{(N)} $ (\ref{Mist1}) and in general case they obey the condition $ \varepsilon_{ij} \neq  \varepsilon_{ji} $. Moreover, the unit sphere centered at a point $ {\bm \xi}^i_0 $ we henceforth designate by
\begin{align}
	\partial \Omega_i^1 \equiv \partial \Omega_i \left({\bm \xi}^i_0; 1 \right) := \lbrace {\bm \xi}_i \in  \Omega_i^-: \xi_i=1 \rbrace\,.
	\label{int2k}	
\end{align} 
%

Here it should be especially noted that performed normalization (\ref{int1}) holds only locally in charts $ \Omega_i^- $ and does not hold in the whole manifold $ \Omega^- $.  Note also that in the interest of readability for all charts $ \Omega_i^- $, transformed according to Eq. (\ref{int1}), we retain the same designations. 

It is expedient now to write the diffusion problem immediately in the local coordinates $ \{ O_i;{\bm \xi}_i\} $. However, sometimes to emphasize the fact that we deal with the field on a manifold coordinateless form $ u_N \left(P \right) $ will be used. Clearly, the basic equation governing exterior steady-state Eq. (\ref{tds1a1}) under partially reflecting  boundary conditions (\ref{tds1a}) and regularity condition at infinity (\ref{rev5i})  with respect to the local concentration $ u_N \left({\bm \xi}_i \right) $ ($ i= \overline{1,N} $) and dimensionless variables (\ref{int1}) reads
\begin{align}
	{\mathbf \nabla}_{{\bm \xi}_i}^2\, u_N \left({\bm \xi}_i \right) =0 \quad \mbox{in} \quad \Omega_i^{-}\,,
	\label{di1}\\
	-\left. \partial _{\xi _i} u_N \left({\bm \xi}_i \right) \right| _{\partial \Omega_i^1}=\kappa_{i}\left[1 -\left.
	u_N \left({\bm \xi}_i \right) \right| _{\partial \Omega_i^1}\right] > 0 \, ,  \label{di2}\\
	\left. u_N \left({\bm \xi}_i \right) \right| _{\xi_i \rightarrow \infty }\rightrightarrows 0 \,.
	\label{di3}
\end{align}	
Hereafter  $ {\bm \nabla}_{{\bm \xi}_i}^2 $ stands for the corresponding dimensionless Laplacian written in local coordinates $ \{ O_i;{\bm \xi}_i\} $. 

Besides we introduce the corresponding dimensionless set of sinks reactivities:  $
\kappa^{(N)}:=\left\{ \kappa_i \right\}_{i=1}^N $ with	$\kappa_i=k_{in}^i/k_{S}^i \in \left[0, \infty \right)$,
where $k_{S}^i=4\pi R_iD $ is the Smoluchowski rate constant (\ref{Zv01}) for the fully absorbing $i$-th sink.

Thus, from a mathematical viewpoint, we deal with the exterior Robin boundary value problem (\ref{di1})-(\ref{di3}) in an unbounded 3D manifold $\Omega ^{-}$ with $N$-connected boundary $\partial\Omega ^{-}$.  An additional point to emphasize is that the boundary conditions (\ref{di2}), being posed on all connected components $ \partial \Omega _i $ of the boundary $ \partial \Omega^- $, are integral, which reflects the important fact that diffusive interaction between sinks are not pairwise additive.

\subsection{The microscopic trapping rate}\label{subsec:rate}

For the theory of diffusion-influenced reactions the most important value is the {\it microscopic trapping rate} for the $ i $-th sink defined as $ k_i := {\Phi_i}/{c_B} $, where $ \Phi_i $ stands for  the total flux of $B$'s on the $i$-th sink surface. 
In turn with the help of known solution $ n_N \left( P \right) $, the total microscopic trapping rate on the reaction surface $\partial \Omega_ i$ can be calculated straightforwardly
\begin{equation}
	k_i =\oint_{\partial \Omega_i}\left. \hat{\mathbf n}_i \cdot \mathbf{j} \right| _{\partial \Omega_i}dS_i \, ,
	\label{Zv06}
\end{equation}
where $ dS_i $ is the $i$-th surface element. 
Then, recasting (\ref{Zv06}) with respect to $ u_N \left({\bm \xi}_i \right) $ in an exterior neighborhood of the $i$-th sink surface, one can easily obtain the desired microscopic trapping rate (\ref{Zv06}) by the surface integral over the unit sphere 
\begin{align}
	k_i (X^{(N)};\kappa^{(N)}) = - \oint_{\partial \Omega_i^1 }\left. \partial _{\xi _i} u_N \right| _{\partial \Omega_i^1}d{\hat{\mathbf \xi}}_i\,.	\label{dm1n}
\end{align}
Hence, for dense enough sink arrays when diffusive interaction effects become important it is necessary to determine the appropriate correction factors to the absorbsion rates on all $ N $ sinks.
Therefore the above rate (\ref{dm1n}) is sought in the form
\begin{align}
	k_i (X^{(N)};\kappa^{(N)}) = k_{S}^i  J_0^i \left(\kappa_{i}\right) J_i(X^{(N)};\kappa^{(N)})\,,
	\label{dm1na}
\end{align}
where the $i$-th dimensionless Collins-Kimball rate given by 
\begin{align}
	J_0^i \left(\kappa_{i}\right)= \kappa_i\slash \left(1 + \kappa_i\right) \quad \mbox{for} \quad i= \overline{1,N} \,.
	\label{dm1na2}
\end{align} 
generalizes one-sink expression. Below, to simplify notations, we omitted corresponding parameters.
Evidently,  $ J_0^i $ is the unperturbed by the diffusive interaction  Collins-Kimball trapping rate for the $ i $-th sink normalized by $k_{S}^i $. 

In representation (\ref{dm1na}) the magnitudes $ J_i $ are commonly called the {\it screening coefficients}~\cite{Cossali20} (rate correction factors~\cite{Labowsky78}). It is absolutely clear that due to dependence on configuration of sinks $ X^{(N)} $ and their reactivities $ \kappa^{(N)} $ screening coefficient $ J_i $ (usually, for simplicity, we shall omit arguments in the function notation) expresses the effect of the diffusive interaction between a given $ i $-th sink and other $ N-1 $ sinks. It emerges from the physical standpoint that the diffusive interaction reduces the reaction rate of either sink compared with that of a single one. Hence, the stronger the diffusive interaction the smaller $ J_i $ is and the following relations hold true 
\begin{align}
	0 < J_i  <1\,, \quad
	\lim_{\varepsilon_{ij} \to 0} J_i  = 1 \quad \mbox{for all} \quad i(\neq j)=\overline{1,N} \,. 	\label{dm1na2b}
\end{align}
Limit (\ref{dm1na2b}) simply means that the diffusive interaction disappears at large separations between sinks.

It is significant that, generally speaking, the microscopic trapping rate (\ref{dm1n}) for $ N \geq 2 $ does not coincide with the reaction rate constant $ k $ defined by Eq. (\ref{Zv00}), which is a macroscopic value~\cite{Bressloff20}. 
So, investigating diffusive interaction effects in systems with $N \geq 2$ sinks we cannot use  the term "rate constant" for $ k_i $ given by Eq. (\ref{dm1n}) any more~\cite{Traytak92}.

\section{Method of solution}\label{sec:method}

The present section is devoted to an overall description of the generalized method of separation of variables, which is a powerful tool for solution of the Robin boundary value problem (\ref{di1})-(\ref{di3}). 

However, before starting to discuss this topic, we briefly highlight the application of the {\it boundary integral equations method}.~\cite{Mikhlin64,Grebenkov19} One can reduce the above diffusion problem (\ref{di1})-(\ref{di3}) to solution of a system of Fredholm boundary integral equations of the II kind. It turned out that this system of integral equations may be reduced to the corresponding infinite system of linear algebraic equations.~\cite{Traytak05} Despite the fact that this method is also applicable to the mixed boundary conditions~\cite{Krutitskii99} its implementation in the diffusion problem at issue  involves a number of disadvantages compared to the approach that would be considered here.

\subsection{On the method of separation of variables}

First recall some well-known mathematical facts on the standard {\it method of separation of variables} (MSV)~\cite{Tikhonov63}. An additional point to emphasize is that present study deals with multiplicatively separable solutions of the diffusion equation only.  

Let us fix a global Cartesian coordinate system  $ \lbrace O; \mathbf{r} \rbrace $ in a bounded or unbounded domain $ \Omega \subset \mathbb{R}^3 $.
\begin{definition}\label{definCan}
	A domain $ \Omega $ is said to be canonical if any solution $ u \left(\mathbf{r} \right) $ to a boundary value problem in $ \Omega $ may be derived by the standard MSV.
\end{definition}
Thus, in other words, there exists an orthogonal functional basis of solutions $ \lbrace \varphi_n \left(\mathbf{r} \right) \rbrace_{n=0}^{\infty} $ (called {\it basis solutions}) to a boundary value problem in a canonical domain  $ \Omega \subset \mathbb{R}^3 $ and a set of real coefficients  $ \lbrace C_n \rbrace_{n=0}^{\infty} $ such that 
\begin{align}
	u \left(\mathbf{r} \right) = \sum_{n=0}^{\infty} C_n \varphi_n \left(\mathbf{r} \right) \quad \mbox{in} \quad \Omega\,, \label{def1}
\end{align}
where series is absolutely and uniformly convergent. For 3D Laplace' equation, e.g., canonical domains are those bounded by surfaces: spheroids, cylinders, cones, paraboloids, and their combinations for interior and corresponding complements in $ \mathbb{R}^3 $ for exterior boundary value problems, respectively.

The specific values of the coefficients $ C_n $ to be determined from the appropriate boundary conditions.
Plainly, both introduced above domains $ \Omega_i $ and their complements $ \Omega_i^- $ ($ i = \overline{1, N} $) are internal and exterior canonical ones for the steady-state diffusion problems, respectively. The relevant functional basises consist of regular and irregular solid harmonics.

The application of the MSV to the well-posed diffusion problems involves:
\begin{itemize}
	\item  application of the local linear superposition principle;
	\item  determination of the basis solutions to the equation in a given canonical 
	domain;	
	\item determination of unknown coefficients from the posed boundary conditions.
\end{itemize}

\subsection{Generalized method of separation of variables}

The {\it generalized method of separation of variables} (GMSV)~\cite{Traytak03,Traytak18,Grebenkov19,Piazza19} naturally stems from the standard method of separation of variables and its incipience go back to the classic works by Maxwell~\cite{Maxwell} and  Basset.~\cite{Basset1887} However, the most clearly this idea was expressed in Rayleigh's seminal paper on the conductivity of heat and electricity in a medium with cylindrical or spherical sinks arranged in a rectangular array~\cite{Rayleigh1892}. Nowadays, this approach is referred to as the Rayleigh multipole method.

Since then, the GMSV had been intensively developed and found numerous applications in electrostatics, hydrodynamics, mechanics, heat transfer, diffraction theory and many other fields. Surprisingly, these advances typically remained "hidden" within each discipline and unknown to the researchers, worked in other fields. Such a parallel development led to multiple "rediscoveries" of the same results in different fields. Wherein, various names were given to the GMSV by different authors.  It is also widely known as the "generalized Fourier method"~\cite{Nikolaev16,Miroshnikov22}. Guz' and Golovchan treated the method as essentially a particular case of the "method of series"~\cite{Guz}. Finally, in micromechanics of heterogeneous materials mainly term "multipole expansion approach" is used~\cite{Kushch13}. 

Note that all above names capture only some features of the method under consideration. As this method relies on the separation of variables in local curvilinear coordinates, following Ivanov, it can be call "generalized method of separation of variables"~\cite{Ivanov}. We have to emphasize that here a particular case of the GMSV for the sinks of the same form (spherical) will be utilized. So we shall not apply the
{\it re-expansion theorems} connecting different basis functions corresponding to the different shapes of sinks.

Among the numerous works devoted to the development of GMSV, the Ivanov's book holds a special place~\cite{Ivanov}.
For the first time Ivanov presented the scheme of the GMSV in full details and, moreover, he derived a number of new and reviewed known then addition theorems for the sets of basis solutions to the Helmholtz equation, written in different curvilinear coordinate systems~\cite{Ivanov}.  Although Ivanov used the method to solve various problems concerning diffraction theory for two particles, but it is clear that the same approach is also valid for finite number of particles. 

\begin{figure}[h]
	\centering
	\includegraphics[scale=0.6]{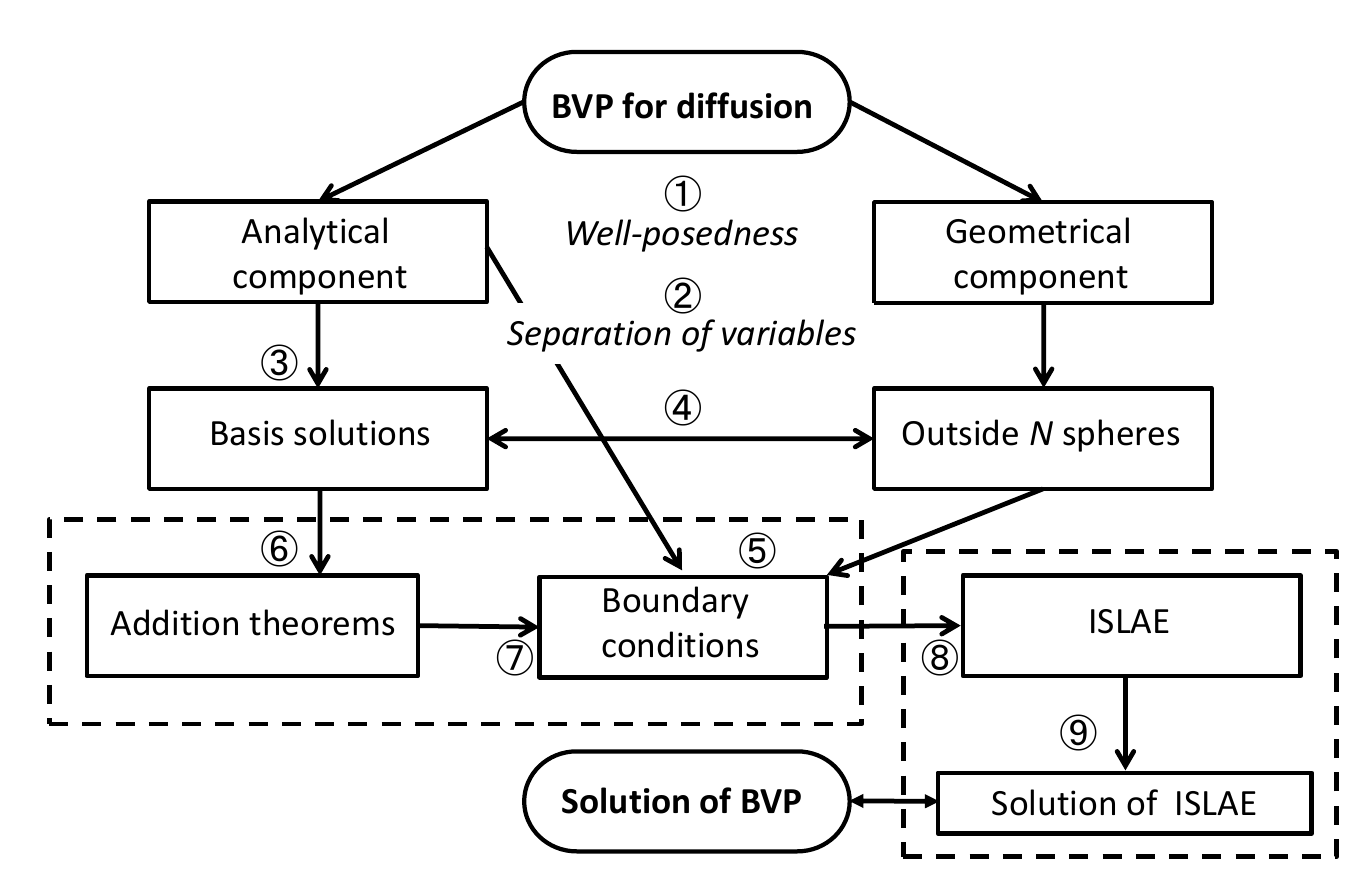}
	\caption{A block diagram of the generalized method of separation of variables for diffusion problems in the manifold $ \Omega^- $}\label{fig5}
\end{figure}

The basic idea of the GMSV consists in reducing the original boundary value problem posed in the configuration manifold $\Omega^-$ to the corresponding $N$ boundary value problems in charts $\Omega^-_i$. These $ N $ problems are coupled through the boundary conditions, prescribed on the whole disconnected boundary $\partial \Omega^-$. Thus, the GMSV is actually the method of separation of variables on the 3D manifold~\cite{Traytak03}.

\subsection{Outline of approach}\label{subsec:methodOut}

The application of the GMSV to the boundary value problem (\ref{di1})-(\ref{di3}) may be formulated as the {\it GMSV algorithm} which in turn can be divided up into following main steps.
\begin{enumerate}
	\item Proof of the well-posedness of the solution to the posed boundary value problem;
	\item  Deconposition of the desired solution in the form of the sum of partial solutions by means of the general linear superposition principle;
	\item  Determination of the appropriate basis solutions to the diffusion boundary value problem posed in the given canonical domains with disconnected boundary;	
	\item  Determination of relevant addition theorems to obtain local regular expansions of partial solutions in a neighborhood of corresponding connected components of bounday;	
	\item  Application of the addition theorems in order to satisfy the boundary conditions; 
	\item  Reduction of the problem to a resolving self-consistent infinite system of linear algebraic equations;
	\item  Solution of the resolving infinite system of linear algebraic equations;
	\item Calculation of the microscopic reaction rate.
\end{enumerate}
In Fig.~\ref{fig5} we provide a full scheme of the GMSV application to a typical  diffusion boundary value problem. One can see that the {\it linearity} of the boundary value problem is one of the primary requirements for the GMSV algorithm as well.

Simple inspection reveals that the fundamental differences between standard MSV and GMSV (see Fig.~\ref{fig5}) are in the blocks depicted by long-dashed lines. Hence, contrary to the MSV the GMSV algorithm requires two supplementary rather non-trivial mathematical tools: the appropriate addition theorems for the basis solutions along with methods for of the resolving infinite system of linear algebraic equations.

\section{The GMSV with irreducible Cartesian tensors}\label{sec:Irreducible}

Previously, by means of GMSV for 3D domains with $ N $ disconnected spherical boundaries, we presented the detailed derivation of the Laplacian Green functions for both the exterior and interior Dirichlet and Robin problems along with the Green function for the relevant conjugate (transmission) boundary value problem~\cite{Grebenkov19}. Some facets of the GMSV applications to the theory of stationary bulk diffusion-influenced reactions on spherical sinks were also studied and discussed in Refs.~\cite{Galanti16a,Grebenkov19,Piazza19}. However, it is important that all analytical and numerical calculations were performed there using the local spherical coordinate systems.

Plainly the general solution algorithm is feasible in two ways: by means of either solid spherical harmonics with respect to a polar spherical coordinate system or, alternatively, with the help of the ICT (for details, see below Appendix~\ref{sec:appendix}). The wide use of solid harmonics in the theory of diffusion-influenced reactions is a matter of common knowledge (see, e.g., Refs.~\cite{Venema89,Tsao01a,Tsao02,Strieder03,Strieder04} and references cited there). In the meantime, studies devoted to similar applications of ICT are only very few~\cite{Kapral79,Beenakker86,Traytak91,Traytak92,Traytak01,Traytak03}. This seems surprisingly, taking into account the facts that ICT formalism has a number important advantages against standard solid harmonics approach (see discussion below in Sec.~\ref{sec:discussion}). 
In this connection, it is particularly remarkable that recent decades investigations have inspired a renewed interest for use of the ICT technique to study a number of problems in systems with $N$-spheres microstructure (\ref{Mist1})~\cite{Hess03,Martynov10,Agre11,Guskov13,Adhikari15,Hess15,Mane16,Bechinger16,Adhikari17,Andrews20,Turk22}. 

Against these recent trends, the lack of new works concerning the applications of the ICT in diffusion-influenced reactions theory is especially noticeable. 

To our knowledge Lebenhaft and Kapral were the first who successfully applied the ICT technique to the diffusion problem (\ref{di1})-(\ref{di3})~\cite{Kapral79}. Subsequently, an explicit form of the GMSV with the aid of ICT was carried out in our paper Ref.~\cite{Traytak91}, where we utilized this form of the GMSV to solve problems on Ostwald ripening, taking into account multipole corrections. It should be stressed that our approach stems from the {\it method of Cartesian ansatz} suggested in 1978 by Schmitz and Felderhof~\cite{Felderhof78}. Although Ref.~\cite{Felderhof78} studied a hydrodynamic problem, authors also treated the Robin boundary value problem for Laplace's equation in the unbounded domain exterior to a sphere. In that regard it is pertinent to cite here above paper: "We solve this problem by making an ansatz in cartesian coordinates, rather than following the usual method of introducing spherical coordinates." 
The Schmitz-Felderhof ansatz means that one is seeking a solution in the form of mixture dependence of spherical radial $ \xi_i $ and corresponding Cartesian local coordinates  $ \xi_{\gamma _m}^{i} $ in every $i$-th chart ($ i = \overline{1, N} $)~\cite{Felderhof78}
\begin{align}
	u_N \left({\bm \xi}_i \right) = \sum\limits_{n=0}^\infty w_n \left({\xi}_i  \right) C_{\gamma _1 \ldots \gamma _n}^i \odot^n \mathop {{\xi
		}_{\gamma _1}^i{\ldots \xi }_{\gamma _n}^i} \limits^{\rule[-.04in]{.01in}{.05in}%
		\rule{.47in}{.01in}\rule[-0.04in]{.01in }{.05in}} \quad \mbox{in}\quad \Omega_i^- \,.   
	\label{ansatz1}
\end{align}
Hereinafter $ \mathop {{\xi
	}_{\gamma _1}^i{\ldots \xi }_{\gamma _n}^i} \limits^{\rule[-.04in]{.01in}{.05in}%
	\rule{.47in}{.01in}\rule[-0.04in]{.01in }{.05in}} $ is the ICT defined by Eq. (\ref{it1v}). The explicit form of functions $ w_n \left({\xi}_i  \right) $ may be found by substituting in Laplace's equation (\ref{di1}) while the unknown tensor coefficients $ C_{\gamma _1 \ldots \gamma _n}^i  $ to be determined from the boundary conditions (\ref{di2}).

Note in passing that so-called {\it self-consistent field method} that is another version of the above method of Cartesian ansatz~\cite{Guskov13}. A method of reflections with the help of ICT used in Ref.~\cite{Chen95} to solve the problem of the spherical particles motion due to a temperature gradient is also worth mentioning here. It is important to mention that some approaches vide used just different modifications of the GMSV in terms of the  above treated version with the help of the ICT~\cite{Kim87,Tsao01}.

In Ref.~\cite{Beenakker86} a scalar version of well-known induced force method has been applied to the theory of diffusion-controlled reactions for the first time. It turns out, however, that both induced forces method and the GMSV by means of the ICT lead to the same second kind ISLAE with respect to unknown tensor constants involved~\cite{Traytak03,Traytak05}.

\section{The addition theorem in terms of the irreducible tensors}\label{sec:addtheor}


Due to the fact that the Laplace equation (\ref{di1}) remains invariant under the action of the group of Euclidean motions~\cite{Miller77} {\it irregular solid harmonics} (see definitions in Appendix~\ref{sec:appendix}) written in local coordinates of any chart $ \Omega_i^- $ ($ \Omega_j^- $) can be recast in terms of {\it regular solid harmonics} written in local coordinates of another chart $ \Omega_j^- $ ($ \Omega_i^- $). 

One of the key point of the GMSV algorithm is the transformation of the solid harmonics under action of subgroup of translations (\ref{adg1c})~\cite{Steinborn73,Piazza19}.
Provided charts $ \Omega_i^- $ and  $ \Omega_j^- $ possess the same form (complements of spheres in the present paper), the appropriate transformation formulas of that kind we shall call {\it irregular to regular translation addition theorems} (I$\to$R TAT). 
\begin{remark}
	On the contrary, in the case when forms of domains $ \Omega_i^- $ and  $ \Omega_j^- $ are different (like the exteriors of a sphere $ \Omega_1^- $ and a prolate spheroid $ \Omega_2^- $ see Ref.~\cite{Traytak18}) corresponding relations it is more expedient to call {\it translation re-expansion theorems}.
\end{remark}

The goal of this section is to present a proof of the I$\to$R TAT within the scope of the GMSV in terms of the ICT.
We note in passing that, strictly speaking, we deal here with scalar translation addition theorems only.

\subsection{The degenerate translation addition theorem}

First let us focus our treatment on a particular case of the I$\to$R TAT , which plays a key role in the subsequent proof of the general TAT.

Assume that $ r > \Vert \mathbf{r}^i_0 \Vert $, where $ r = \Vert {\mathbf r}\Vert $ is the distance between the global origin $ O $ and a point $ P $ (see Fig.~\ref{fig1}).

Thereby, we can write the known Taylor's expansion
\begin{align} 
	\frac{1}{\Vert {\mathbf{r}}_i\Vert} =	\frac{1}{\Vert {\mathbf r} - \mathbf{r}^i_0 \Vert} = \sum_{k=0}^{\infty} \frac{(-1)^k}{k!} \left( \mathbf{r}^i_0 \cdot {\mathbf\nabla} \right)^k \left( \frac{1}{r} \right)\,.  \label{ret1a1}
\end{align}
Clearly, series (\ref{ret1a1}) converges absolutely and uniformly for all $ r > 0 $~\cite{Tikhonov63}. 

Function $ 1/r $ is often termed the {\it generating function}.
Traditionally, Eq. (\ref{ret1a1}) is referred to as the {\it multipoles expansion of the (global) fundamental solution} of the Laplacian at the pole point $ \mathbf{r}^i_0 $ in powers of $ r^{-1} $ (see, e.g., Ref.~\cite{Tikhonov63}). 

For the further consideration, it is convenient to recast expansion (\ref{ret1a1}) in terms of the ICT defined by Eq. (\ref{it1v})~\cite{Hess15}. This fact allows us to formulate so-called {\it the degenerate translation addition theorem} for solid harmonics 
\begin{theorem}\label{theorem1}
	For all $ r > 0 $ the fundamental solution of the Laplacian by means of the ICT may be expanded in an absolutely and uniformly convergent series
	\begin{align}	
		\frac{1}{\Vert \mathbf{r}-\mathbf{r}^i_0\Vert}	=\sum\limits_{k=0}^\infty \omega_k r^{-(2k+1)}\mathop {\hspace{0.2cm}{r}_{\gamma _1}{\ldots r}_{\gamma _k}}\limits^{%
			\rule[-.04in]{.01in}{.05in}\rule{.45in}{.01in}\rule[-0.04in]{.01in
			}{.05in}}  \odot^k \mathop {r_{\gamma _1}^{0i}{ \ldots r}_{\gamma _k}^{0i}}\limits^{%
			\rule[-.04in]{.01in}{.05in}\rule{.45in}{.01in}\rule[-0.04in]{.01in
			}{.05in}} \,, \label{it6r}
	\end{align}
	where $ \omega_k:=  \left( 2k-1\right)!!/k!$,  $r_{\gamma _\nu }$ and $ r_{\gamma _m}^{0i} $ are Cartesian coordinates of vectors $ \mathbf r $ and $ \mathbf{r}^i_0 $, respectively.
\end{theorem}

This form of multipole expansion is very broadly utilized across various scalar fields of the physical and chemical sciences. Moreover, (\ref{it6r}) plays a key role in derivation of the translation addition theorem and that is why it is also called the degenerate addition theorem.

\subsection{The translation addition theorem}

It has been known that I$\to$R TAT holds true for solid harmonics $ \psi^{\pm}_{nm} \left( {\mathbf r} \right) $ defined in polar spherical coordinate systems (see, e.g., Ref.~\cite{Grebenkov19} and references therein). Here we shall present a derivation of the similar I$\to$R TAT in terms of ICT $ \bm{X}^{\pm}_{n} $ (see Appendix~\ref{sec:appendix}), which was given in Refs.~\cite{Traytak91,Traytak92} without a proof.

To calculate the local concentration $ u_N\left( {P}\right) $  at
any given point $ P \in \Omega ^{-} \subset \Omega ^{-}_i \cap \Omega ^{-}_j$ with position vector $ {\mathbf r} $ we consider functions connected with every coordinate chart $\left( \Omega ^{-}_i,\varphi _i\right) $: $u^{\prime }\left( \mathbf{r}_i\right) =u_N\left( {P}\right)$ for $ \mathbf{r}_i=\varphi _i\left( {\mathbf{r}}\right)$ ($ i=\overline{1, N}$) (see~\ref{subsec:addtheorDef}). 
In any other chart $\left( \Omega_j ^{-},\varphi _j\right) $ at the same
point $P \in \Omega ^{-}$ we have alternatively
\begin{eqnarray}
	u^{\prime \prime }\left( \mathbf{r}_j\right) =u^{\prime }\left( \mathbf{r}%
	_i\right) =u_N\left( {\mathbf{r}}\right)\,, \label{man1} \\	
	u^{\prime \prime }\left( \varphi_{ij}\left({\mathbf{r}}_i\right)\right) = u^{\prime }\left( \varphi_{ji}\left({\mathbf{r}}_j\right)  \right)=u_N\left( {\mathbf{r}}\right)\, .
	\label{man2}
\end{eqnarray}

For our problem the overlap map $ \varphi_{ij} $ is a simple translation (\ref{it2}) so let us start with an important general
\begin{definition}
	Any two given local curvilinear coordinate systems $\{ O_i;{\bm\zeta}_i\}$ and $\{ O_j;{\bm\zeta}_j\}$ are said to be consistent if there exists a translation (\ref{adg1c}) with vector $ {\mathbf a} $ such that $ \mathbf {\bm\zeta}_j =T_{\mathbf a}{\bm\zeta}_i = {\bm\zeta}_i + {\mathbf a}$. 
\end{definition}
This definition immediately brings to
\begin{proposition}\label{CC1}
	Any two local Cartesian coordinate systems $\{ O_i;{\mathbf r}_i\}$ and $\{ O_j;{\mathbf r}_j\}$ are consistent if and only if: (a) they are positively oriented; and (b) all their axes are pairwise parallel with respect to the global Cartesian coordinates.
\end{proposition}
One can see that local Cartesian coordinates $ \lbrace O_{i}; {{\mathbf r}_{i}} \rbrace $  and $ \lbrace O_{j}; {\mathbf r}_{j} \rbrace $  depicted in Fig.~\ref{fig1} are consistent.  
Note in passing that some authors prefer to use term "parallel coordinate systems" rather than "consistent coordinate systems"~\cite{Steinborn73}.

\begin{definition}
	A translation addition theorem for a given smooth function $u: \Omega ^{-} \to \mathbb{R}_+ $ is a formula expressing its value $u(\mathbf{r} + \mathbf{a})$ in terms of values $u(\mathbf{r})$ and $u(\mathbf{a})$,  where $ \mathbf{r}, \mathbf{r} + \mathbf{a} \in \Omega ^{-} $, and also their derivatives.
\end{definition}
\begin{definition}
	The translation addition theorem I$\to$R yields an explicit formula for converting an irregular ICT $ \bm{X}_{n}^{-} $ given in one origin $ O_j$ into a local expansion with respect to the corresponding regular ICT $ \bm{X}_{n}^{+} $ about a shifted origin $ O_i$.
\end{definition}
The meaning of the addition theorems is to derive corresponding addition formulas. 

Now, we can state and prove our main result in this section.
\begin{theorem}\label{theorem2}
	The irregular to regular translation addition theorem for the irreducible Cartesian tensors  holds true 
	\begin{eqnarray}
		r_j^{-\left( 2n+1\right) }\mathop {{r }_{\gamma _1}^j {\ldots r }_{\gamma
				_n}^j}\limits^{\rule[-.04in]{.01in}{.05in}\rule{.5in}{.01in}%
			\rule[-0.04in]{.01in }{.05in}}  
		=\sum\limits_{k=0}^\infty U^{ij}_{\gamma _1 \ldots \gamma _n\mu
			_1  \ldots \mu _k}\left( {\widehat{\mathbf{L}}}_{ij}\right) \odot^k \mathop {{r}_{\mu _1}^i {\ldots r }_{{\mu _k}}^i}\limits^{\rule[-.04in]{.01in}{.05in}\rule{.5in}{.01in}%
			\rule[-0.04in]{.01in}{.05in}}\,, \label{it3a}
	\end{eqnarray}	
	where the transformation matrix elements read
	\begin{eqnarray*}
		U^{ij}_{\gamma _1 \ldots \gamma _n\mu
			_1  \ldots \mu _k}\left( {\widehat{\mathbf{L}}}_{ij}\right) =\sigma
		_{kn}  L_{ij}^{-\left(n+k+1 \right)} \Lambda _{\gamma _1 \ldots \gamma _n\mu
			_1  \ldots \mu _k}\left( {\widehat{\mathbf{L}}}_{ij}\right)\, , \\
		\sigma _{kn}={\left( -1\right) ^n}\frac{\left[ 2\left( k+n\right) -1\right]
			!!}{{k!\left( 2n-1\right) !!}}\,, \\
		\Lambda _{\gamma _1 \ldots \gamma _n\mu _1 \ldots \mu _k}\left( {\widehat{\mathbf{L}}}_{ij}\right) =\mathop {\hat{L}_{\gamma _1}^{ij}{\ldots}%
			\hat{L}_{\gamma _n}^{ij}\hat{L}_{\mu _1}^{ij}{\ldots}\hat{L}_{\mu _k}^{ij}}\limits^{%
			\rule[-.04in]{.01in}{.05in}\rule{1.04in}{.01in}\rule[-0.04in]{.01in
			}{.05in}}\,, \quad
		\hat{L}_{\mu _m}^{ij}= L_{\mu _m}^{ij} / L_{ij}\, , \quad m \in \mathbb{N}\,
	\end{eqnarray*}
	wherein series (\ref{it3a}) converges uniformly and absolutely if inequalities (\ref{int2e}) hold true.
\end{theorem}
\begin{proof} 
	Rewriting Eq. (\ref{it2}) in the relevant local Cartesian coordinates of $j$-th and $i$-th spheres we have 
	\begin{align}
		r_{\gamma _\nu }^j=r_{\gamma _\nu }^i+L_{\gamma _\nu }^{ij}\,,
		\label{it3b}
	\end{align}
	where $L_{\gamma _\nu }^{ij}$ are coordinates of vector $ {\mathbf L}_{ij} $ ($\gamma _\nu =\overline{1,3}$). Linear dependence (\ref{it3b}) leads to the evident relations $\partial
	_{r_{\gamma _\nu }^j}=\partial _{r_{\gamma _\nu }^i}$. 	
	Using them in Eq. (\ref{it1v}), one readily obtains
	\begin{eqnarray}
		r_j^{-\left( 2n+1\right) }\mathop {r_{\gamma _1}^j {\ldots r}_{\gamma _n}^j}%
		\limits^{\rule[-.04in]{.01in}{.05in}\rule{.45in}{.01in}%
			\rule[-0.04in]{.01in
			}{.05in}}=\frac{\left( -1\right) ^n}{\left( 2n-1\right) !!}\partial
		_{r^i_{\gamma _1}}\ldots\partial _{r^i_{\gamma _n}}\left( \frac 1{\left\| \mathbf{r}%
			_i-\mathbf{L}_{ij}\right\| }\right)  \nonumber	
	\end{eqnarray}	
	or, taking into account that $\partial _{r_{\gamma _\nu }^i}=-\partial_{L_{\gamma _\nu }^{ij}} $,
	one has the relation
	\begin{eqnarray}
		r_j^{-\left( 2n+1\right) }\mathop {r_{\gamma _1}^j{\ldots r}_{\gamma _n}^j}%
		\limits^{\rule[-.04in]{.01in}{.05in}\rule{.45in}{.01in}%
			\rule[-0.04in]{.01in}{.05in}}  
		=\frac 1{\left( 2n-1\right) !!}\partial _{L_{\gamma
				_1}^{ij}}\ldots \partial _{L_{\gamma _n}^{ij}}\left( \frac 1{\left\| \mathbf{r}_i-%
			\mathbf{L}_{ij}\right\| }\right) \,.  \label{it5}
	\end{eqnarray}	
	Theorem \ref{theorem1} yields 
	\begin{eqnarray}
		\frac{1}{\Vert \mathbf{r}_i-\mathbf{L}_{ij}\Vert}	=\sum\limits_{k=0}^\infty \frac{\left( -1\right) ^k}{k!}  
		\partial _{L_{\gamma _1}^{ij}} \ldots \partial _{L_{\gamma
				_k}^{ij}}\left(\frac{1}{L_{ij}}
		\right) \odot^k \mathop {r_{\gamma _1}^i{ \ldots r}_{\gamma _k}^i}\limits^{%
			\rule[-.04in]{.01in}{.05in}\rule{.5in}{.01in}\rule[-0.04in]{.01in
			}{.05in}} \,. \label{it6rn}
	\end{eqnarray}	
	Substituting this expansion into Eq. (\ref{it5}) and taking into consideration the uniform convergence of the series (\ref{it6rn}) we arrive at the required translation addition theorem (\ref{it3a}). 
\end{proof}
Note that, following Miller, the transformation matrix $ U^{ij}_{\gamma _1 \ldots \gamma _n\mu
	_1  \ldots \mu _k} $ is commonly termed the {\it mixed-basis matrix}~\cite{Miller77}.

For further investigation of the problem (\ref{di1})-(\ref{di3}) it is expedient to recast expression (\ref{it3a}) in a dimensionless form
\begin{eqnarray}
	\xi _j^{-\left( 2n+1\right) }\mathop {{\xi }_{\gamma _1}^j {\ldots\xi }_{\gamma
			_n}^j}\limits^{\rule[-.04in]{.01in}{.05in}\rule{.5in}{.01in}%
		\rule[-0.04in]{.01in }{.05in}}  
	=\sum\limits_{k=0}^\infty \widetilde{U}^{ij}_{\gamma _1 \ldots \gamma _n\mu
		_1  \ldots \mu _k}\odot^k \mathop {{\xi }_{\mu _1}^i {\ldots \xi }_{{\mu _k}}^i}\limits^{\rule[-.04in]{.01in}{.05in}\rule{.5in}{.01in}%
		\rule[-0.04in]{.01in}{.05in}}\,, \label{di4b}
\end{eqnarray}	
where dimensionless mixed-basis matrix elements are
$$
\widetilde{U}^{ij}_{\gamma _1 \ldots \gamma _n\mu
	_1  \ldots \mu _k}:=\sigma
_{kn} \varepsilon _{ij}^{n+1}\varepsilon _{ji}^k\Lambda _{\gamma _1 \ldots \gamma _n\mu
	_1  \ldots \mu _k}\left({\widehat{\mathbf{L}}}_{ij}\right)\, .$$
Particularly dimensionless form of the degenerate translation addition theorem (\ref{it6r}) reads
\begin{align}
	\xi _j^{-1 }=\sum\limits_{k=0}^\infty \widetilde{U}^{ij}_{\mu
		_1  \ldots \mu _k} \odot^k \mathop {{\xi }_{\mu _1}^i {\ldots \xi }_{{\mu _k}}^i}\limits^{\rule[-.04in]{.01in}{.05in}\rule{.5in}{.01in}%
		\rule[-0.04in]{.01in
		}{.05in}}\,. \label{di4c}	
\end{align}

The above presented proof of the I$\to$R TAT in terms of ICT seems to be the simplest among previously known theorems for solid harmonics in terms of polar spherical coordinates~\cite{Caola,Weniger85}. It is important to stress that using connection between spherical solid harmonics and ICT (see Appendix~\ref{sec:appendixCon}) the proved addition theorem (\ref{it3a}) can be reduced  to that for the spherical solid harmonics written in a spherical coordinate system. 

Finally, note that to solve diffusion problems under proper mixed boundary conditions~\cite{Traytak07}, when intrinsic rates are functions on the angular local coordinates $ \kappa_i \left(\theta_i, \phi_i \right) $, we need the {\it rotational addition theorem for the ICT}. The latter theorem may be proved in entirely similar way.

\section{Solution to the problem}\label{sec:solution}

In this section, we implement the GMSV algorithm described in Sec.~\ref{sec:method} using ICT technique. Moreover we shall develop here the version of the GMSV elaborated previously in Refs.~\cite{Traytak91,Traytak92,Traytak01}.

\subsection{The GMSV algorithm implementation}

For clarity sake we organize our treatment in according to the method outline~\ref{subsec:methodOut} distinguishing   8 steps inherent the GMSV.

\textbf{Step 1.} The mixed boundary value problem (\ref{di1})-(\ref{di3}) is well posed and has a unique classical solution~\cite{Courant53,Tikhonov63,Krutitskii99}.

\textbf{Step 2.} First, we can apply superposition principle. Decompose the general solution
$ u_N \left(P \right) $ of above diffusion problem (\ref{di1})-(\ref{di3}) with respect to the $i$-th sink, considering the set of {\it partial solutions} to Eq. (\ref{di1}), which are defined on the charts: $ \left\{ u_i \left({\bm \xi}_i \right) \right\}_{i=1}^N $ ($u_i: \Omega_i^- \to \left(0,1\right] $). Thus, in the local Cartesian coordinates $\{ O_i; {\bm \xi}_i\} $ we can write
\begin{align}
	u_N \left({\bm \xi}_i \right) = u_i \left({\bm \xi}_i  \right) +  \sum\limits_{j\left( \not =i\right) =1}^N u_j \left({\bm \xi}_j  \right)\quad \text{in} \quad \Omega^-_i \,. \label{di2a}
\end{align}
Here, by virtue of Eqs. (\ref{it2}) and (\ref{int2e}), one has evident relations:  $ \varepsilon_{ij}{\bm \xi}_j  = \varepsilon_{ji}{\bm \xi}_i  + {\widehat{\mathbf{L}}}_{ij} $ (see Fig.~\ref{fig1}). From the physical standpoint representation (\ref{di2a}) means that to find the desired reaction rate on the $i$-th sink we need to know only local behavior of the field $ u_N \left(P \right) $ in a vicinity of this sink.

Clearly, by definition functions $ u_i \left({\bm \xi}_i \right) $ are harmonic in charts $ \Omega ^{-}_i $, i.e.
\begin{eqnarray}
	{{\bm \nabla}}_{{\bm \xi}_i}^2 u_i\left({\bm \xi}_i \right) = 0 \,, \label{int2}\\
	\left. u_i\left({\bm \xi}_i \right) \right|_{\xi _i \rightarrow \infty } \rightrightarrows 0\, .\label{int3}
\end{eqnarray}
Therefore, the original diffusion boundary value problem (\ref{di1})-(\ref{di3}) in the manifold $\Omega^-$ (intersection of all charts $ \Omega ^{-}_i $) is reduced to $ N $ coupled problems for $u_i\left({\bm \xi}_i \right)$ in simpler charts $\Omega^-_i$, with the advantages of the local Cartesian coordinates $\{ O_i; {\bm \xi}_i\} $.  

Alternatively, for inhomogeneous partially reflecting boundary conditions (\ref{di2}) in an exterior $ \epsilon $-neighborhood of the boundary $ \partial \Omega_i $: $ \Omega ^{-}_i\left(\epsilon \right):= \{{\bm \xi}_i: 1 < \xi_i  < 1 +\epsilon\} \subset \Omega ^{-}_i $  it is convenent to recast solution $ u_N \left({\bm \xi}_i \right) $ as
\begin{align}
	u_N \left({\bm \xi}_i \right) = u_{\left( 0\right) }^i {\left( \xi_i \right) } + \delta u_i \left({\bm \xi}_i  \right)\, \quad \text{in} \quad \Omega ^{-}_i\left(\epsilon \right) \,,   \label{di2b}
\end{align}
where $ u_{\left( 0\right) }^i {\left( \xi_i \right) } $ is the unperturbed one-sink solution and $ \delta u_i \left({\bm \xi}_i  \right) $ is a perturbation due to influence of each from other $ N-1 $ sinks for $  {j\left( \not =i\right) =\overline{1, N}}$. Owing to Eq. (\ref{di2}) we assume that these solutions obey inhomogeneous and homogeneous partially reflecting boundary conditions, respectively
\begin{eqnarray}	
	- \left. \left( \partial _{\xi _i}  u_{\left( 0\right) }^i -\kappa_{i} u_{\left( 0\right) }^i\right)\right| _{\partial \Omega_i^1}=\kappa_{i}  \,, \label{di3a}\\
	\left. \left( \partial _{\xi _i} \delta u_i -\kappa_{i}\delta u_i\right)\right| _{\partial \Omega_i^1}=0\,.\label{di3b}
\end{eqnarray}
Therewith the unperturbed by diffusive interaction solution
\begin{align}
	u_{\left( 0\right) }^i {\left( \xi_i \right) }= J_0^i  {\xi_i}^{-1}\, \label{di4a0}
\end{align}
obeys inhomogeneous boundary condition (\ref{di3a}). Here $ J_0^i $  is the $i$-th Collins-Kimball rate (\ref{dm1na2}).  

One important point to emphasis is that Eqs. (\ref{di2a}) and (\ref{di2b}) are global and local representations of the same function defined on the manifold:  $u_N: \Omega^- \to \left(0,1\right] $.

\textbf{Step 3.} It can be shown that tensors fields introduced by Eqs. (\ref{it1vd1}) and (\ref{it1vd2}) are families of irregular  $ \lbrace \mathbf{X}_{n}^{-}\left({\bm \xi}_i  \right) \rbrace_{n=0}^{\infty} $ and regular $ \lbrace \mathbf{X}_{n}^{+}\left({\bm \xi}_i  \right) \rbrace_{n=0}^{\infty} $ solid Cartesian harmonics, which form desired basis functions (a complete and orthogonal sets) in domains $ \Omega^-_i $ and $ \Omega^+_i $, respectively.
Thus, we can represent a partial solution $ u_i \left({\bm \xi}_i  \right) $ with the help of the irregular ICT, satisfying regularity condition at infinity 
\begin{align} 
	u_i \left({\bm \xi}_i  \right) = \sum\limits_{n=0}^\infty \xi _i^{-\left( 2n+1\right) }
	A_{\gamma _1...\gamma _n}^i \odot^n  \mathop {{\xi}%
		_{\gamma _1}^i{\ldots \xi}_{\gamma _n}^i}\limits^{\rule[-.04in]{.01in}{.05in}%
		\rule{.5in}{.01in}\rule[-0.04in]{.01in
		}{.05in}}\quad \mbox{in} \quad \Omega^-_i \,. \label{di4}
\end{align}

On the other hand, in an exterior neighborhood of $i$-th sink boundary $\partial\Omega_i$ we can represent the perturbation function in Eq. (\ref{di2a}) as a series with respect to the regular ICT 
\begin{align} 
	\sum\limits_{j\left( \not =i\right) =1}^N u_j \left({\bm \xi}_j  \right) = \sum\limits_{n=0}^\infty B_{\gamma _1 \ldots \gamma _n}^i \odot^n \mathop {{\xi
		}_{\gamma _1}^i{\ldots \xi }_{\gamma _n}^i} \limits^{\rule[-.04in]{.01in}{.05in}%
		\rule{.47in}{.01in}\rule[-0.04in]{.01in }{.05in}}\,. \label{di4a}
\end{align}
In Eqs. (\ref{di4}) and (\ref{di4a}) $A_{\gamma _1 \ldots \gamma _n}^i$ and $B_{\gamma _1 \ldots \gamma _n}^i$ are tensor coefficients to be determined from the boundary conditions (\ref{di2}).

\textbf{Step 4.} Theorem \ref{theorem2} in the dimensionless form (\ref{di4b}) gives the required translation addition theorem to tackle the diffusion problem (\ref{di1})-(\ref{di3}). 

\textbf{Step 5.} Now, by means of dimensionless translation addition theorem (\ref{di4b}), we satisfy the Robin boundary conditions (\ref{di2}). Hence,  taking advantage of the ICT linear independence, we can express the coefficients $A_{\gamma _1 \ldots \gamma _n}^i$ in terms of $B_{\gamma_1 \ldots \gamma _n}^i$ 
\begin{eqnarray}
	A_0^i=J_0^i \left(1-B_0^i\right) \quad \mbox{for} \quad  n=0 \, , \label{di4bc}\\
	A_{\gamma _1\ldots \gamma _n}^i=\frac{n-\kappa_{i}}{1+n+\kappa_{i}}B_{\gamma _1 \ldots \gamma _n}^i%
	\quad \mbox{for} \quad n \in {\mathbb N} \, . \label{di4cc}
\end{eqnarray}

With the aid of these connections formula (\ref{di2b}) for the perturbed diffusion field in an exterior neighborhood of the surface $ \partial \Omega_i $ ($ i = \overline{1, N} $) yields
\begin{eqnarray}
	\delta u_i \left({\bm \xi}_i  \right)  = \sum\limits_{k=0}^\infty\left[ 1+\left( \frac{k-\kappa_{i}}{1+k+\kappa_{i}}\right) \xi
	_i^{-\left( 2k+1\right)
	}\right] B_{\gamma
		_1 \ldots \gamma _k}^i \odot^k \mathop {{\xi}_{\gamma _1}^i{ \ldots \xi}_{\gamma _k}^i}\limits^{%
		\rule[-.04in]{.01in}{.05in}\rule{.47in}{.01in}\rule[-0.04in]{.01in
		}{.05in}}\, . \label{di6} 
\end{eqnarray}
One can can easily verify that this solution obeys homogeneous partially reflecting boundary condition (\ref{di3b}).

And, in addition, using Eq. (\ref{di4}) under relationships (\ref{di4bc}) and (\ref{di4cc}) for domaind $ \Omega^-_j $ we immediately write down an auxiliary relation
\begin{eqnarray}
	\sum\limits_{j\left( \not =i\right) =1}^N u_j \left({\bm \xi}_j  \right) = \sum\limits_{j \left( \not =i\right) =1}^N u_{\left( 0\right) }^j {\left( \xi_j \right) }  \nonumber\\	
	+\sum\limits_{j \left( \not =i\right) =1}^N 
	\sum\limits_{k=0}^\infty \left( \frac{k-\kappa_{j}}{1+k+\kappa_{j}}\right) \xi _j^{-\left(
		2k+1\right)} B_{\gamma
		_1 \ldots \gamma _k}^j \odot^k\mathop {{\xi}_{\gamma _1}^j{\ldots\xi}_{\gamma _k}^j}\limits^{%
		\rule[-.04in]{.01in}{.05in}\rule{.45in}{.01in}\rule[-0.04in]{.01in
		}{.05in}}  \,. \qquad \label{di7}  
\end{eqnarray}
It is noteworthy that all terms here satisfy the regularity condition at infinity (\ref{di3}) as opposed to the local expression (\ref{di6}).

\textbf{Step 6.} Now in the local Cartesian coordinates $\{ O_i; {\bm \xi}_i\} $ one can carry out a {\it self-consistent procedure} to find unknown tensor coefficients $B_{\gamma _1 \ldots \gamma _k}^i$ appearing in formulas (\ref{di6}) and (\ref{di7}).
First we apply the dimensionless translation addition theorem (\ref{di4b}) including its degenerate form (\ref{di4c}) to the right-hand side of Eq. (\ref{di7}) in order to recast it in the local coordinates $\{ O_i; {\bm \xi}_i\} $. Then with allowance made for expression (\ref{di4a}) and orthogonality property of the ICT on the unit spheres $ \partial \Omega_i^1 $ (\ref{Asur2}) we derive the required self-consistent {\it infinite system of linear algebraic equations} (ISLAE) of the II kind with respect to unknown tensor coefficients $B_{\gamma _1 \ldots \gamma _k}^i$. We should note that these infinite systems are commonly referred to as the {\it resolving ISLAE}~\cite{Kushch13,Traytak19}.

According to the statement of the diffusion problem given in Sec. \ref{sec:state} we should distinguish three particular cases of the resolving ISLAE. 

\textbf{(a)} For partialy reflecting sinks (when $0 < \kappa_{i} < \infty $) the resolving ISLAE is
\begin{eqnarray}
	B_{\gamma _1 \ldots \gamma _k}^i = B_{\gamma _1 \ldots \gamma _k}^{i\left(
		0\right)} + \sum\limits_{j \left( \neq i \right) =1}^N \sum\limits_{l=0}^\infty \left( \frac{k-\kappa_{j}}{1+k+\kappa_{j}}%
	\right) \nonumber\\	
	\times \widetilde{U}^{ij}_{\gamma _1 \ldots \gamma _k\mu
		_1  \ldots \mu _l} \odot^l  B_{\mu _1 \ldots \mu _l}^j \,, \quad  k = \overline{0, \infty}\,. \label{di9w} 
\end{eqnarray}
Hereinafter for short we denoted
\begin{align}
	B_{\gamma _1 \ldots \gamma _k}^{i\left( 0\right) }= \sum\limits_{j \left( \neq i \right) =1}^N J
	_0^j \widetilde{U}^{ij}_{\gamma
		_1  \ldots \gamma _k} \,. 	\label{di9w0}  
\end{align}
Another two important ISLAE are special cases of (\ref{di9w}).

\textbf{(b)} In case of fully absorbing sinks (when $ \kappa_{i} \to \infty $) the resolving ISLAE is simplified to
\begin{eqnarray}
	B_{\gamma _1 \ldots \gamma _k}^i = B_{\gamma _1 \ldots \gamma _k}^{i\left(
		0\right)} - \sum\limits_{j \left( \neq i \right) =1}^N \sum\limits_{l=0}^\infty  	
	\widetilde{U}^{ij}_{\gamma _1 \ldots \gamma _k\mu
		_1  \ldots \mu _l} \odot^l  B_{\mu _1 \ldots \mu _l}^j \,,  \label{di9S}
\end{eqnarray}
where $ k = \overline{0, \infty} $.

\textbf{(c)} For fully reflecting (when $ \kappa_{i} \to 0 $) the resolving ISLAE (\ref{di9w}) obviously yields
\begin{align}
	B_{\gamma _1 \ldots \gamma _k}^i =  \left( \frac{k}{1+k}\right)\sum\limits_{j \left( \neq i \right) =1}^N \sum\limits_{l=0}^\infty  
	\widetilde{U}^{ij}_{\gamma _1 \ldots \gamma _k\mu
		_1  \ldots \mu _l} \odot^l  B_{\mu _1 \ldots \mu _l}^j \,,  \label{di9a}
\end{align}
where $ k = \overline{1, \infty} $.

\textbf{Step 7.} Consider here the most general case of arrays with partialy reflecting sinks. In other words to find $ B_{\gamma _1 \ldots \gamma _k}^i $ the system (\ref{di9w}) requires the inversion of an corresponding infinite-dimensional matrix~\cite{Galanti16a,Piazza19}. So we deal with a typical problem which may be solved by methods of functional analysis~\cite{Kantorovich,Akilov}. For instance, provided the ISLAE (\ref{di9w}) is of {\it normal Poincare-Koch type} it implies that the {\it Fredholm-Hilbert alternative} holds true~\cite{Kantorovich}. Hence, a unique solution of (\ref{di9w}) exists and it may be found by the {\it method of reduction} to any degree of accuracy~\cite{Traytak19}. The latter means that the resolving ISLAE (\ref{di9w}) should be truncated and then inverted obtained finite system numerically to get approximations of the coefficients $  B_{\gamma _1 \ldots \gamma _k}^i$. Under some conditions the resolving ISLAE (\ref{di9w}) possesses solution which may be obtained by means of simple iterations~\cite{Galanti16a} (see also discussion below in Sec.~\ref{sec:convergence}).  Obviously the zero and the first iterations for $ B_0^i $ are
\begin{eqnarray}
	B_{0}^{i\left( 0\right)}= \sum\limits_{j \left( \neq i \right) =1}^N J
	_0^j \widetilde{U}^{ij}_{0}\,,   \label{di10a}\\	
	B_0^{i(1)} = B_{0}^{i\left(
		0\right)}-\sum\limits_{j \left( \neq i \right) =1}^N J
	_0^j \widetilde{U}^{ij}_{0}B_{0}^{j\left( 0\right)} - \sum\limits_{j \left( \neq i \right) =1}^N\sum\limits_{l=1}^\infty J
	_0^j  
	\widetilde{U}^{ij}_{\mu
		_1  \ldots \mu _l} \odot^l  B_{\mu _1 \ldots \mu _l}^{j\left(
		0\right)} \,.   \label{di10b} 	
\end{eqnarray}
Here and below $ \widetilde{U}^{ij}_{0} $ is the monopole approximation to the  mixed-basis matrix.	
Interestingly that in the course of solving the ISLAE (\ref{di9w}) by simple iterations angular dependence arises starting from terms ${\mathcal{O}}\left( \varepsilon ^4 \right) $.

\textbf{Step 8.} Once the local concentration $ u_N \left(P\right) $ is found one can deduce the total flux of diffusing particles $ B $ onto the $ i $-th sink. At that one should utilize the local representation (\ref{di2b}) of function $ u_N \left(P\right)$. Due to orthogonality condition for the ICT (\ref{Asur2}) explicit surface integration on the unit sphere in the general formula (\ref{dm1n}) will yield the desired rates (\ref{dm1na}).
\begin{remark}
	Quite apparently, functions $ \mathop {{\hat{\xi}}_{\gamma _1}^i\ldots {\hat{\xi}}_{\gamma _n}^i}\limits^{%
		\rule[-.04in]{.01in}{.05in}\rule{.47in}{.01in}\rule[-0.04in]{.01in
		}{.05in}} $ are harmonics polynomianls defined on the $ i $-th unit sphere $ \partial \Omega_i^1 $ and, in accordance with the Weierstrass approximation theorem, 
	any function $w_i {\left( {\hat{\bm \xi}}_i\right)} \in C(\partial \Omega_i^1) $ may be expanded in the series 
	\begin{align}
		w_i \left({{\hat{\bm \xi}}}_i\right)
		= \sum\limits_{n=0}^\infty \alpha _{\gamma _1 \ldots \gamma _n}^i \odot^n \mathop {{\hat{\xi}}_{\gamma _1}^i\ldots {\hat{\xi}}_{\gamma _n}^i}\limits^{%
			\rule[-.04in]{.01in}{.05in}\rule{.45in}{.01in}\rule[-0.04in]{.01in
			}{.05in}}\,, \label{sur2}
	\end{align}
	where $ \alpha _{\gamma _1 \ldots \gamma _n}^i $ are some constant tensor coefficients to be determined by means of the orthogonality condition (\ref{Asur2}).  One can see that expression (\ref{sur2}) makes it possible to apply the ICT method being considered to the diffusion problems under general Robin boundary conditions (\ref{tds1a}).  	 
\end{remark}

\subsection{Main results obtained by using of the ICT technique}\label{subsec:Results}

Within this subsection we present our main results, which are conveniently to formulate in a form of two theorems. 

Thus, for an arbitrary sinks microstructure $ X^{(N)} $ and reactivities set $ \kappa^{(N)} $ we have proved the following assertions.

\textbf{(a)} On global representation of the solution $ u_N \left(P \right) $.
\begin{theorem}\label{theorem3ab}
	Using the local Cartesian coordinates $\{ O_j; {\bm \xi}_j\} $ ($ j=\overline{1,N} $) the solution $ u_N \left(P \right) $ ($ P \in \Omega ^{-} $, see Fig.~\ref{fig1}) to the exterior Robin boundary value problem (\ref{di1})-(\ref{di3}) may be repesented as follows:
	\begin{align}
		u_N \left(P \right) = \sum\limits_{j =1}^N & \left[  J_0^j \left(1-B_0^j\right){\xi_j}^{-1} + 
		\sum\limits_{k=0}^\infty \left( \frac{k-\kappa_{j}}{1+k+\kappa_{j}}\right)  \right. \nonumber\\
		&\left. \quad  \times  \xi _j^{-\left(
			2k+1\right)} B_{\gamma
			_1 \ldots \gamma _k}^j \odot^k\mathop {{\xi}_{\gamma _1}^j{\ldots\xi}_{\gamma _k}^j}\limits^{%
			\rule[-.04in]{.01in}{.05in}\rule{.45in}{.01in}\rule[-0.04in]{.01in
			}{.05in}}  \right] \,, \label{mr1}
	\end{align}
	where tensor coefficients $B_{\gamma _1 \ldots \gamma _k}^i$ ($k = \overline{0, \infty}$) are solution of the resolving ISLAE (\ref{di9w}).
\end{theorem}

\textbf{(b)} On local representation of the solution $ u_N \left(P \right) $.
\begin{theorem}\label{theorem3a}
	Solution $ u_N \left({\bm \xi}_i \right) $ to the exterior Robin boundary value problem (\ref{di1})-(\ref{di3}) in a small exterior $ \epsilon $-neighborhood of the unit sphere $ \partial \Omega_i^1 $ with respect to the local Cartesian coordinates $\{ O_i; {\bm \xi}_i\} $ reads
	\begin{eqnarray}
		u_N \left({\bm \xi}_i \right) = J_0^i {\xi_i}^{-1} + \sum\limits_{k=0}^\infty\left[ 1+\left( \frac{k-\kappa_{i}}{1+k+\kappa_{i}}\right) \xi
		_i^{-\left( 2k+1\right)
		}\right]  \nonumber\\	
		\times  B_{\gamma
			_1 \ldots \gamma _k}^i \odot^k \mathop {{\xi}_{\gamma _1}^i{ \ldots \xi}_{\gamma _k}^i}\limits^{%
			\rule[-.04in]{.01in}{.05in}\rule{.47in}{.01in}\rule[-0.04in]{.01in
			}{.05in}} \quad \mbox{in} \quad \Omega ^{-}_i\left(\epsilon \right)\quad i = \overline{1, N} \,; \quad \label{di6u} 
	\end{eqnarray} 
	where tensor coefficients $B_{\gamma _1 \ldots \gamma _k}^i$ ($k = \overline{0, \infty}$) are solution of the resolving ISLAE (\ref{di9w}).
\end{theorem}
So, if we know the solution $ u_N \left({\bm \xi}_i \right) $ in an exterior neighborhood adjacent to the $ i $-th reaction surface (\ref{di6u}), we can calculate the $B$-particles flux on the $ i $-th sink and, therefore, the corresponding absorbing rate and, in this way, screening coefficient (\ref{dm1na}). Thus, taking advantage of Theorem~\ref{theorem3a}, one can readily prove an important
\begin{corollary}\label{corollary0}
	For the exterior Robin boundary value problem (\ref{di1})-(\ref{di3}) the screening coefficient in Eq. (\ref{dm1na}) has the form   
	\begin{align}	
		J_i = 1 - B_0^i \quad \mbox{and} \quad 0 \leq B_0^i <1 \,,	\label{d11}
	\end{align}
	where coefficients $ B_0^i $  are determined by means of solution to the resolving ISLAE (\ref{di9w}).
\end{corollary}

It is worth noting here that in Ref.~\cite{Krasovitov91} Krasovitov thoroughly reproducing the notation, derivations and even misprints of our paper~\cite{Traytak91}. However, he has been the first who kept all terms comprising $ B_{\gamma _1...\gamma _n}^i $ ($ n \in \mathbb N $) in the series with respect to orthogonal system of ICT after integration over the unit sphere (\ref{dm1n}) (see detailed discussion of these points in Ref.~\cite{Traytak01}).

\subsection{\label{sec:Steady}Monopole approximation}

As we already noted above in Sec.~\ref{sec:DiffInter} the monopole approximation is the simplest tool to study the problem under consideration and, therefore, the most commonly used in various applications. Here we dwell on this approximation briefly.

It is commonly accepted the following
\begin{definition}\label{MOA}
	Monopole approximation corresponds to truncations of series (\ref{di4}) and (\ref{di4a}) for any partial solutions $ u_i \left({\bm \xi}_i  \right) $ ($ i=\overline{1,N} $) such that: $ A_0^i,  B_0^i \neq 0$  and  $ A_{\gamma _1 \ldots \gamma _k}^i, B_{\gamma _1 \ldots \gamma _k}^i  \equiv 0 $ for all $ k \in \mathbb{N} $.
\end{definition}

It is notable that one can see directly from the resolving ISLAE (\ref{di9a}) that $ B_{0}^i=0 $ and, therefore, the fully reflecting obstacles cannot be described in the monopole approximation. Besides,  for this case due to Eq. (\ref{d11}) we have $ J_i = 1 $ as should be from the physical standpoint.

For definiteness sake  below we will consider the case of fully diffusion-controlled reactions (i.e., the case \textbf{(b)} in \textbf{step 6.}) only.

Utilizing system (\ref{di9S}) one infers that for the monopole approximation, values of the total flux 
$J_i$ may be determined by the much simpler finite $ N \times N $ system of equations:
\begin{eqnarray}
	J_i=1-\sum\limits_{j =1}^N \widetilde{V}^{ij}_0  J_j\,,	\quad i=\overline{1,N}\,;  \label{di11} \\
	\widetilde{V}^{ij}_0 := \left\{
	\begin{array}{ll}
		\widetilde{U}^{ij}_0 & \quad \mbox{if} \quad j \neq i \, , \\
		0 & \quad \mbox{if} \quad  i=j \, .
	\end{array}
	\right.   \label{di11v} 
\end{eqnarray}
Therefore, expressions (\ref{di2a}) and (\ref{d11}) are simplified to 
\begin{eqnarray}
	u_N=u_{\left( 0\right) }^i+B_0^i\left( 1-\frac 1{\xi_i}\right)\, ,  \label{di12a}\\
	J _i = 1-B_0^i \,.  \label{di12b} 
\end{eqnarray}
In the important particular case at $ N=2 $ system (\ref{di11}) gives 
\begin{align}
	J_i=1- \varepsilon_{ij} J_j\,, \quad (i \neq j)=1, 2 \label{di13a}
\end{align}
with the first iteration
\begin{align}
	J_i=1- \varepsilon_{ij} \,, \quad (i \neq j)=1, 2\,. \label{di13b}
\end{align}
It is evident that system (\ref{di13a}) has the exact solution
\begin{align}
	J_i= \frac{1-\varepsilon_{ij}}{1-\varepsilon_{ji}\varepsilon_{ij}}  \,, 	\quad (i \neq j)=1, 2\,. \label{di13c}
\end{align}
Note in passing that sometimes in literature Eq. (\ref{di13c}) is called the monopole approximation, however, the monopole approximation is given by Eq. (\ref{di13b}). Besides Eq. (\ref{di13c}) may be also obtained by the reflection method for the monopole approximation~\cite{Traytak92}.

\subsection{General remarks on the obtained solution}

In Sec.~\ref{sec:DiffInter} we already stated that exact solution to the diffusion boundary value problem (\ref{di1})-(\ref{di3}) is possible only for two fully absorbing sinks by means of bispherical coordinates~\cite{Carstens70,Deutch77,Dubinko89,Green06}. Using the bispherical coordinates for the problem under partially reflecting boundary conditions (\ref{di2}) leads to rather cumbersome recurrence relations and therefore does not provide any advantages compared to other methods. On the contrary, for our approach the kind of boundary conditions does not matter. Moreover, the calculation of the diffusion field $ u_N \left(P\right) $  from the boundary value problem (\ref{di1})-(\ref{di3}) in general case of $ N $ sinks microstructure (\ref{Mist1}), when $ N >2 $ is a far too complicated task to be solved analytically in closed form.

A comparison of the general resolving ISLAE (\ref{di9w}) with that (\ref{di9S}) for fully absorbing sinks shows the weakening of the diffusive interaction effects in arrays comprising the partialy reflecting sinks. Clearly the diffusive interaction vanishes for the fully reflecting obstacles. 

An $ N $-sink calculation reveals that the pure diffusion field of $B$-particle around any $i$-th sink $ u_{\left( 0\right) }^i {\left( \xi_i \right) } $ (\ref{di4a0}) is influenced by a contribution mediated by all other $ N-1$ sinks (\ref{di6u}), leading, in it turn, to the screening effects similar to that in an electrostatics~\cite{Deutch76}.

We emphasize again that discussed in Sec.~\ref{sec:state} mixed diffusion boundary value problems with three kinds of obstacles may be treated straightforwardly by the GMSV in terms of the ICT.
We will bring a pattern of numerical calculations of the screening coefficients $ J_i$ (\ref{d11}) for the Dirichlet-Robin proper mixed problem at $ N=3 $ below in Sec.~\ref{sec:examples}.

The GMSV in terms of the ICT may be applied to catalytically activated diffusion-influenced trapping reactions occuring in host media comprising arrays of both finitely many sinks and immobile catalytic domains~\cite{Oshanin05}. 

Finally, it is worth noting that stationary diffusion equation under so-called conjugate boundary conditions~\cite{Mityushev19,Grebenkov19,Traytak19} may be also tackled with above method quite similarly.

\section{Convergence analysis of the GMSV}\label{sec:convergence}

Thus, in principle, we can describe theoretically the trapping of $B$-reactants for any finite number $N$ of spherical sinks. However, for some sink arrays one faces the difficulties with convergence of iterative solution to the resolving ISLAE (\ref{di9w}). Alternative numerical solution of the aforementioned truncated system of equations corresponding to (\ref{di9w}) also leads to serious difficulties. 
So, estimations on validity of the GMSV with the help of the ICT for solution of the diffusion problem at issue naturally arises. This question appeared to be rather complicated and, therefore, we discuss it very shortly. Throughout this section we shall deal with the problem (\ref{di1})-(\ref{di3}) for the case of fully diffusion-controlled reactions.

\subsection{Relations with method of reflections}

The point is that according to the weak maximum principle for harmonic functions we have inequality
\begin{align}
	\max_{ P \in \overline{\Omega}^-}  \vert u_N \left( P \right) \vert \leq \max_{P \in \partial\Omega^-} \vert u_N \left( P \right) \vert \,.  \label{13r}
\end{align}
Let a harmonic function $ v_N \left(P\right) $ be an approximation to the exact solution $ u_N \left(P\right) $ of the diffusion problem (\ref{di1})-(\ref{di3}). It immediately follows from inequality (\ref{13r}) that the maximum absolute error of the residual for any approximation $ \vert v_N \left(P\right) -  u_N \left(P\right) \vert $ is achieved on the boundary of the configuration manifold $ \partial\Omega^- $. So, it is a simple matter to find a priori the upper estimate of the error for the residual of a given approximation in the configuration  manifold $ {\Omega}^- $.

As will become clear below, our consideration concerning convergence of the GMSV procedure it is expedient to start with some results on the method of reflections (in mathematical literature it is also known as the generalized Schwarz’s method)~\cite{Smirnov64,Traytak06,Salomon21}.  For two fully absorbing sinks the reflection method leads to a series which reveals slow down convergence as sinks become closer up to the contact, still being convergent (see discussion below). However, even in case of three sinks convergence of the method is not guaranteed at all intersink separations~\cite{Traytak06}. Indeed, studing similar problem on the hydrodynamic interaction in the Stokes flow past three spheres, Kim drew attention to the fact that even in three spheres case convergence of the method of reflections can be violated~\cite{Kim87}. We do not intend to discuss this subject here and so we refer the interested reader to  Ref.~\cite{Salomon21}. 

According to the method of reflections one shold construct a functional sequence of relevant approximations $\left\{ v_m\left( P\right) \right\}_{m=0}^{\infty} $ so-called {\it reflections} (see p. 638 of Ref.~\cite{Smirnov64} and  Ref.~\cite{Traytak06} for details) such that this sequence converges in $\Omega ^{-}$ uniformly to the exact solution 
\begin{align}
	v_m\left( P\right) \rightrightarrows u_N \left(P\right)  \quad \text{as} \quad  m \to \infty\,.    \label{13ra}
\end{align}
In Ref.~\cite{Traytak06} we proved the following theorem on the conditions for uniformly convergence of the method of reflections (\ref{13ra}).
\begin{theorem}\label{theorem4}
	For the relationship (\ref{13ra}) to hold the following sufficient (a) and necessary (b) conditions must be satisfied
	\begin{eqnarray}
		\mbox{(a)} \quad M_N :=  \max_{i=\overline{1,N}}\sum\limits_{k\left( \neq
			i\right) =1}^N\left( \frac{\varepsilon_{ik}}{1-\varepsilon_{ki}
		}\right) <1 \,, \label{14r}\\
		\mbox{(b)} \quad m_N :=   \min_{i=\overline{1,N}}\sum\limits_{k\left( \neq
			i\right) =1}^N\left( \frac{\varepsilon_{ik}}{1+\varepsilon_{ki}
		}\right) <1 \,. \label{15r}
	\end{eqnarray}
\end{theorem}
It is important to highlight that currently it is believed these inequalities are still the best ones to estimate convergence of the method of reflections~\cite{Salomon21}.

Plainly, Theorem \ref{theorem4} implies the following two assertions.
\begin{corollary}
	For sinks of identical radius $ R_i =R $ ($ i= \overline{1,N} $) conditions (\ref{14r}) and (\ref{15r})  are simplified to 
	\begin{eqnarray}
		\mbox{(a)} \quad M_N =   \max_{i=\overline{1,N}}\sum\limits_{k\left( \neq
			i\right) =1}^N\left( \frac{L_{ik}}{R} -1 \right)^{-1} < 1 \,, \label{14e}\\
		\mbox{(b)} \quad m_N =  \min_{i=\overline{1,N}}\sum\limits_{k\left( \neq
			i\right) =1}^N\left( \frac{L_{ik}}{R} + 1 \right)^{-1} < 1 \,. \label{15e}
	\end{eqnarray}
\end{corollary}
\begin{corollary}\label{corollary1}
	For the monopole approximation conditions (\ref{14r}) and (\ref{15r})  are reduced to 
	\begin{align}
		\mbox{(a)} \quad M_N = \max_{i=\overline{1,N}}\sum\limits_{k\left( \neq
			i\right) =1}^N\varepsilon_{ik} < 1  \,, \quad  \mbox{(b)} \quad   m_N =  \min_{i=\overline{1,N}}\sum\limits_{k\left( \neq
			i\right) =1}^N\varepsilon_{ik} < 1\,. \label{137r}
	\end{align}
\end{corollary}
One can see that above condition (a) (\ref{137r})
coincides with well-known sufficient condition for the convergence of the iterative solution of the system (\ref{di11}): 
\begin{align}
	\Vert \widetilde{V}^{ij}_0 \Vert _{\infty}:= \max_{i=\overline{1,N}}\sum_{k=1}^{N}\widetilde{V}^{ik}_0<1\,, \label{137v}
\end{align}
where $ \widetilde{V}^{ij}_0 $ is the hollow matrix (\ref{di11v}).

Moreover, previously in Ref.~\cite{Traytak03} we have shown that this approach is tightly connected with the corresponding iterative solution within the scope of the GMSV and proved
\begin{theorem}\label{theorem5}
	Solution to the resolving ISLAE (\ref{di9S}) by simple-iteration method is equivalent to the method of reflections.
\end{theorem}
In this way Theorem \ref{theorem5} immediately leads to the following remarkable result
\begin{corollary}\label{corollary1}
	Conditions (\ref{14r}) and (\ref{15r}) may be used to estimate the convergence of the simple-iteration method when applied to the resolving ISLAE (\ref{di9S}).
\end{corollary}

Solvability of the resolving ISLAE by simple-iteration method is very important because of its simplicity and, due to this fact, considerable current use for numerical implementation of the approach at issue~\cite{Galanti16a}.

\subsection{Some estimates of the convergence}

Estimates of the convergence for symmetric configuration of sinks are the most simple. First consider  two examples of axisymmetric configurations.

\textbf{Two sinks of different radii.}

For two sink arrays we shall always denote $L = L_{12}=L_{21} $.

Consider first two sinks supposing, for definiteness, that $ R_1 >  R_2 $. 
In this particular case quite apparently the sufficient condition of convergence (\ref{14r}) is satisfied and it yields inequality $  M_2 = R_{1}/ \left( 1 - R_{2}\right) < 1 $. Hence we have 
$ h = L - \left( R_1 + R_2\right) > 0 $, where $ h $ is the distance between sink surfaces.
In its turn, the latter inequality coincides with non-overlapping sinks condition (\ref{Zv00b}).

Then, by Theorem \ref{theorem4} and Corollary~\ref{corollary1} we arrive at the analog of the known Goluzin theorem~\cite{Smirnov64}.
\begin{theorem}\label{theorem8}
	For any two disjoint fully absorbing sinks the resolving ISLAE (\ref{di9S}) may be solved by the method of simple iterations.
\end{theorem}

To show how one can directly apply the Corollary~\ref{corollary1}, as examples consider arrays of $N$ fully absorbing sinks of identical radius $ R_i =R $ ($ i = \overline{1,N} $).  

\textbf{Sinks of equal radii are located along a straight line.}

In addition let us, for simplicity sake, assume that all $N$ sinks are equally spaced at spacing $ a $ along their line of centers (see Fig.~\ref{fig6} (a)). We have performed calculations at $ N=100 $. It turned out that the simple-iteration method using to solve the resolving ISLAE (\ref{di9S}) is not valid when $  R/a < 4.9 $, and when $  R/a < 9.4 $ its applicability requires additional justification.

\begin{figure}[t!]
	\centering
	{\includegraphics[scale=0.7]{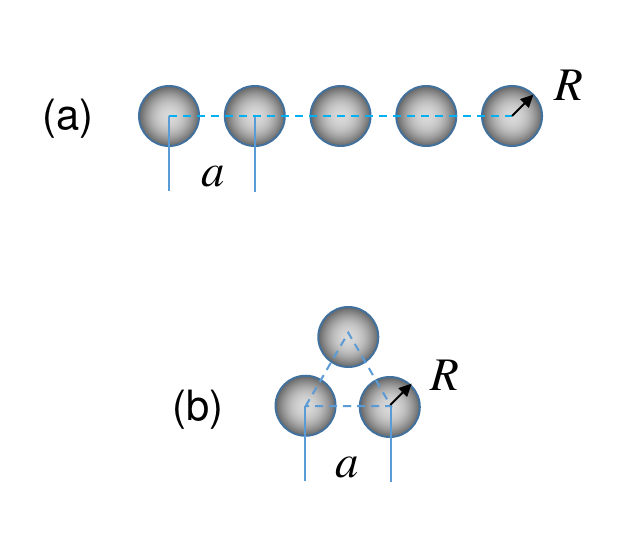}}
	\caption{The array geometry of fully absorbing sinks located: (a) along a straight line; (b) at the apices of an equilateral triangle}\label{fig6}		
\end{figure}

\textbf{A triangular arrangement of equal sinks.}

Diffusion-controlled reactions in the three-sink arrays, particularly equal sinks arranged in an equilateral triangle  have been the subject of numerous studies~\cite{Deutch76,Labowsky78,Traytak92,McCammon13}. 

So, consider now the array containing three identical sinks arranged so their centers are at the apices of an equilateral triangle with sides of length $ a $ (see Fig.~\ref{fig6} (b)). One can easy verify that sufficient condition for the convergence of the simple-iteration method (\ref{14e}) fails when $ 1/3 < R/a < 1/2 $.
This is entirely consistent with the conclusion by Kim in Ref.~\cite{Kim87}, where were found numeically that for three sphere system relevant reflections series will diverge when $ 0.46 \lesssim  R/a < 0.5 $ and it becomes convergent for $ R/a < 0.4 $. Moreover, Kim stated there: "From a rigorous viewpoint, it is not entirely clear at this point whether this is because of an inherent limitation of a method based on the expansion in $ R/a $ (method of reflections) or of because some error in the analysis."

For further results concerning the convergence of the simple-iteration method in the problem at issue, interested readers are encouraged to refer to the recent works Refs.~\cite{Mityushev19,Salomon21}.

\section{Illustrative examples}\label{sec:examples}

In this section as illustrations of the above GMSV by means of the ICT we numerically evaluated the diffusive interaction for a few simple examples of many-sink arrays.

\subsection{Iso-concentration surfaces for two sinks}

To be specific, assume that the sinks are fully absorbing and start from the simplest, and yet, important case of two sinks.
\begin{figure}[t!]
	\centering
	\includegraphics[scale=0.5]{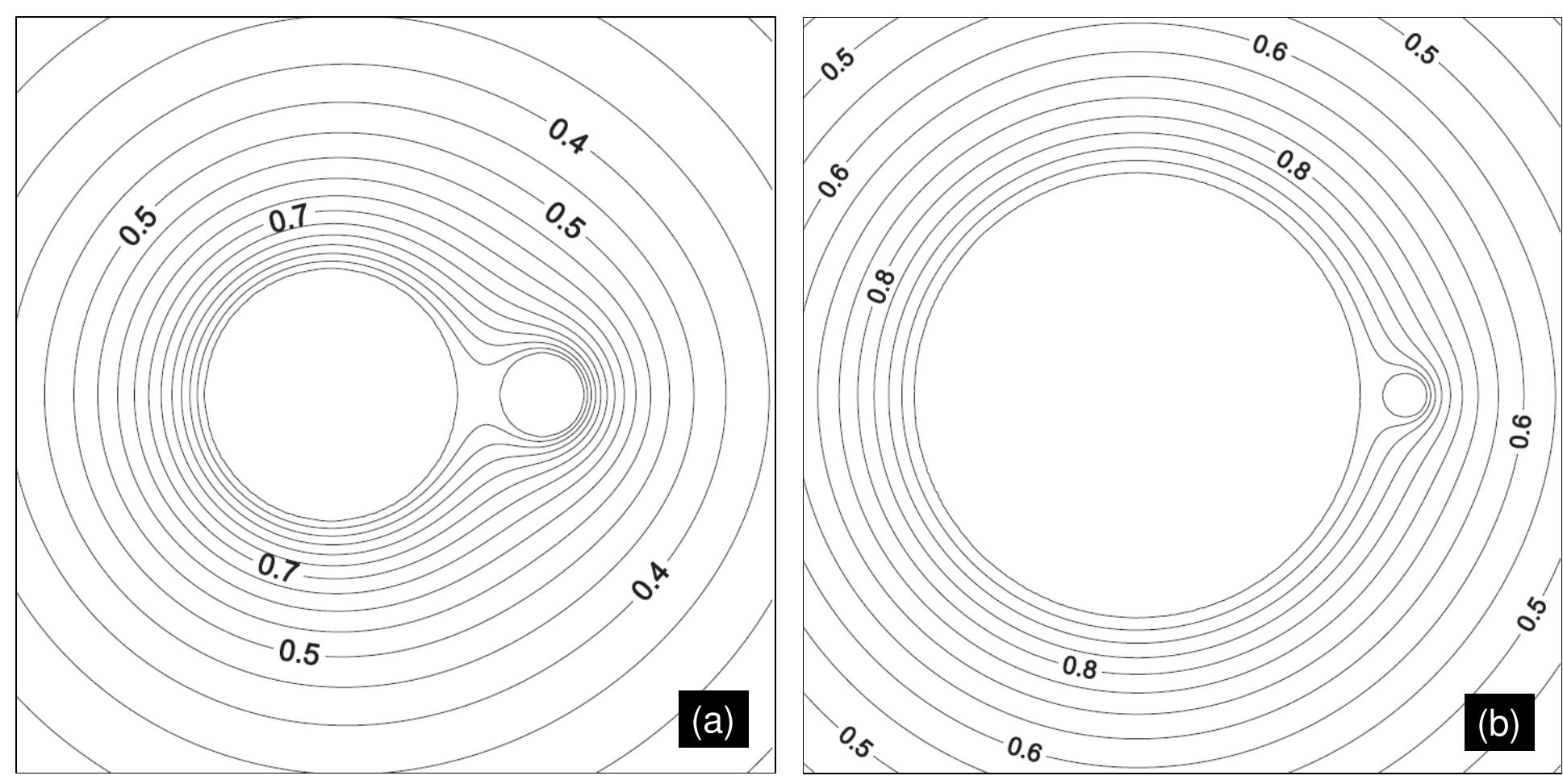}
	\caption{Cross section through the level surfaces of $ u_2 \left(P\right)$ for  decrement step $ \vartriangle h = 0.05 $ ($ l_0 = 20 $): (a) radius $R^{\ast}_{1}=3$ and the sink spacing $L^{\ast}=5$; (b) radius $R^{\ast}_{1}=10$ and the sink spacing $L^{\ast}=12$} \label{fig7}
\end{figure}

\begin{figure}[t!]
	\centering
	\includegraphics[scale=0.5]{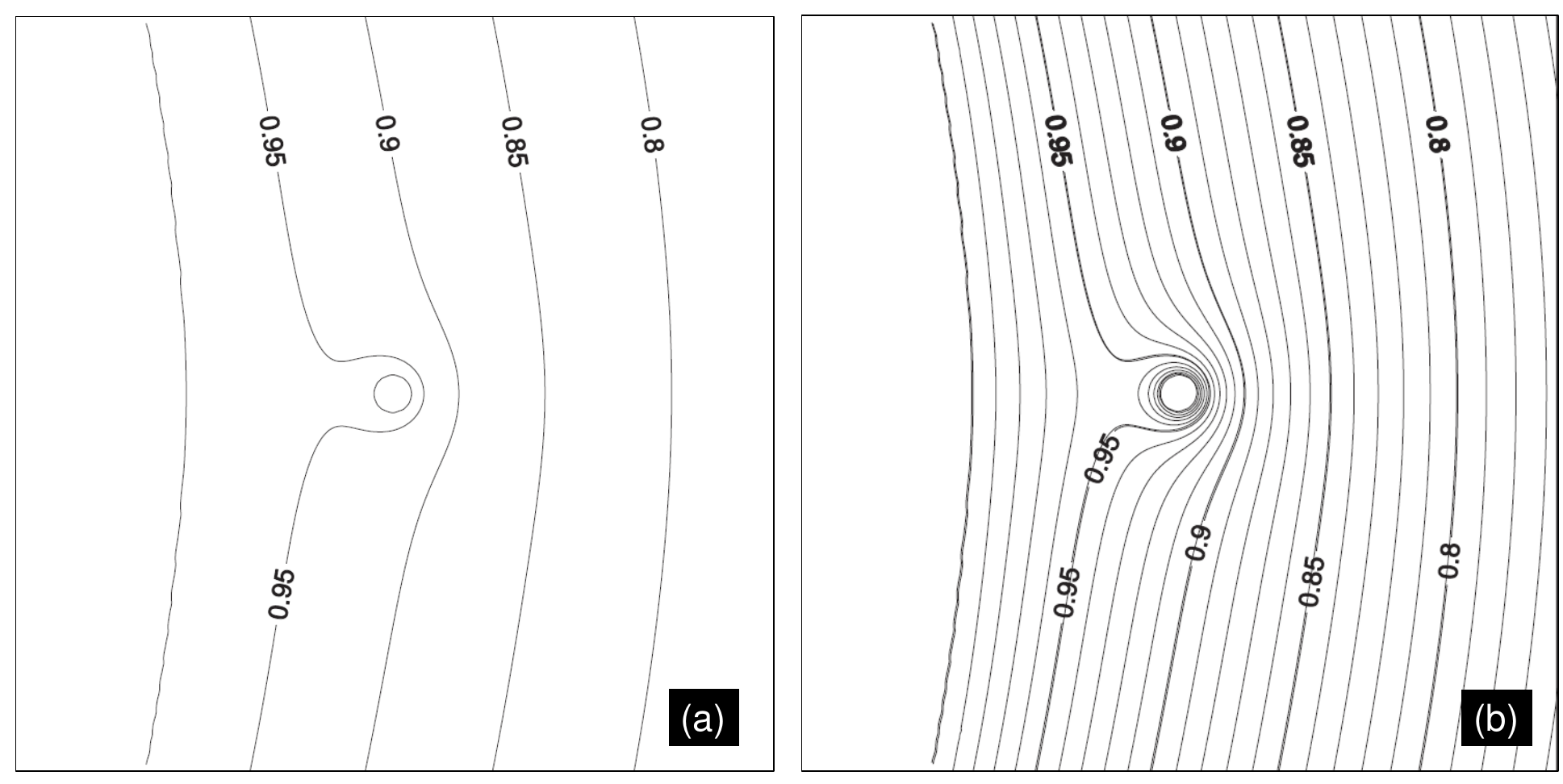}
	\caption{Cross section through the level surfaces of $ u_2 \left(P\right)$ for radius $R^{\ast}_{1}=100$ and the sink spacing $L^{\ast}=111$: (a) decrement step $ \vartriangle h = 0.05 $ ($ l_0 = 20 $); (b) decrement step $ \vartriangle h = 0.01 $ ($ l_0 = 100 $).}
	\label{fig8}
\end{figure}

By means of Theorem \ref{theorem3ab} the local concentration field may be provided in the entire manifold $ \Omega^- $. For this purpose let us introduce the concentration level surfaces: $ S \left( h \right) \subset \Omega^- $ such that  $ \left.  u_2 \left( P \right) \right| _{S} = h \in \left(0, 1 \right) $. As usually for discretization we perform uniform partition of the concentration range $ \left(0, 1 \right) $ with a fixed step value $ \vartriangle h $, such that $1$ correspond to the reaction surfaces and $0$ to infinity: $ h_l = l \vartriangle h$ for $l = \overline{0,l_0}$, where $ l_0=1/\vartriangle h $.

Therewith, for concreteness, we carry out numerical calculations assuming that $ R_1 > R_2 $. 
It is convenient to consider dependence of the local concentration  $ u_2 \left( P\right)$ on two dimensionless geometrical parameters in units of the smaller sink: reduced sink 1 radius $ R^{\ast}_{1}:=R_1/R_{2} $ and the sink spacing $L^{\ast}=L/R_2$.

Calculated cross sections through the appropriate level surfaces have been plotted in the order of decreasing of the hight $h$. In  Fig.~\ref{fig7} and (a) of Fig.~\ref{fig8} we used  decrement step $ \vartriangle h = 0.05 $, whereas calculations in plot (b) of Fig.~\ref{fig8} were performed at $ \vartriangle h = 0.01 $.

Analysis of above numerical results for two fully absorbing sinks allows us to formulate general
\begin{proposition}\label{CC2}
	There exists a number $ h_e \in \left(0, 1 \right) $ such that one can treat the domain 
	$\Omega _e:=\left\{ P: u_N\left( P\right)
	< h_e\right\} $ as an effective fully absorbing sink with respect to the scaled concentration $u^{\ast}_N\left( P\right):=  u_N\left( P\right) \slash h_e$ given on the configuration manifold $\Omega^{-} _e:=\left\{ P: u^{\ast}_N\left( P\right)
	> 1\right\} $.
\end{proposition}

For instance, at $N=2$ Fig.~\ref{fig8} shows that one can take $ h_e = 0.95 $ to get snowman-shaped effective fully absorbing sink with the configuration manifold $  \Omega_e^-= \left\{ P: u_2\left( P\right) 
> 0.95\right\} $.

\subsection{Calculations of the screening coefficients}

Within this subsection to calculate the screening coefficients due to siffusive interaction we shall use the Corollary~\ref{corollary0} of the Theorem~\ref{theorem3a} on local representation of the solution.

For two fully absorbing sinks of different sizes numerical estimations show that the corresponding screening coefficients obey inequality $ J_1 \geq J_2 $ at all sink-sink center separations $ L > R_1 + R_2 $, wherein $ J_1 \to 1 $ as $ R_2 \to 0 $.  Furthermore, there is no diffusive interaction between sinks and any number $ N_p $ ($ 1 \leq N_p \leq N$) of isolated points. Really for the isolated point sinks $ R_j = \varepsilon_{ij}=0 $ and, therefore, corresponding mixed-basis matrix elements in the resolving ISLAE (\ref{di9w}) are equal to zero.

\textbf{Two equal fully absorbing sinks.}

Again assume that both sinks are fully absorbing.
We chose this case taking into account the fact that iteration solution of the resolving ISLAE (\ref{di9S}) converges up to the contact of sinks at $ N=2 $~\cite{Traytak06}.

For simplicity sake consider here two equal sinks $R_1=R_2=R$, so we have: $ \varepsilon_{21}  = \varepsilon_{12} = \epsilon  \in (0, 0.5) $ and for screening coefficients $J_1\left(\epsilon \right)=J_2\left(\epsilon \right)=J\left(\epsilon \right)$. Solving the ISLAE (\ref{di9S}) by simple iterations as $ \epsilon \rightarrow 0+ $ we find:
\begin{eqnarray}
	J\left(\epsilon \right) = J_{a}\left(\epsilon \right) + {\cal O}\left({\epsilon}^6\right) \,, \label{ex0}\\ 
	J_{a}\left(\epsilon \right) = 1-\epsilon + {\epsilon}^2 - {\epsilon}^3 + 2{\epsilon}^4 - 3 {\epsilon}^5 \,. \label{ex0a}
\end{eqnarray} 
It is noteworthy that angular dependence in (\ref{ex0a}) arises starting from the terms $ {\cal O}\left({\epsilon}^4\right) $. 

To show that $ J_{a}\left(\epsilon \right) $ very nearly approximates to the exact screening coefficient we need error analysis.  we quantify the difference between the solutions of these test problems. 
For this purpose let us estimate the above approximation {\it relative error bound} defined as
$$ \delta_m := \sup_{\epsilon \in (0, 0.5)}\delta \left(\epsilon \right) =  \lim_{\epsilon \rightarrow 0.5-}\delta \left(\epsilon \right)\,, $$
where $ \delta \left(\epsilon \right):=  \vert J \left(\epsilon \right) - J_{a} \left(\epsilon \right) \vert /J \left(\epsilon \right) $ is the relative error.

In the case under consideration with the aid of bispherical coordinate system Samson and Deutch obtained the exact reaction rate~\cite{Deutch77}
\begin{align}
	J\left(\epsilon \right)=\sum\limits_{n=1}^\infty \left( -1\right) ^{1+n}\frac{\sinh
		\left( \mu_0 \right) }{\sinh \left( n\mu_0 \right) },
	\label{ex0b}
\end{align}
where 
$ \mu_0 :=\cosh ^{-1}\left(  1/{2\epsilon }\right) $.
For the contact, i.e. when $\epsilon \rightarrow 0.5-$ (or $%
\mu_0 \rightarrow 0+$) we have fully absorbing dumbbell system and formula (\ref{ex0b}) yields the exact limit
$$ \lim_{\epsilon \rightarrow 0.5-}J\left(\epsilon \right) =\ln 2  $$ and, therefore, we get desired upper bound for the relative error
\begin{align}
	\delta \left(\epsilon \right) < \delta_m = 1- \frac{21}{32 \ln 2} \approx 0.053\,. \label{ex0c}
\end{align}
One can see that relative error decreases rather rapidly when $ {\epsilon} $ decreases, e.g., $ \delta \left(0.4 \right) \approx 0.015 $, $ \delta \left(0.3 \right) \approx 0.003 $ and besides $ J \left(\epsilon \right) > J_{a} \left(\epsilon \right)$ for all $ \epsilon \in (0, 0.5) $.

\begin{figure}[t!]
	\centering
	\includegraphics[width=0.65\textwidth]{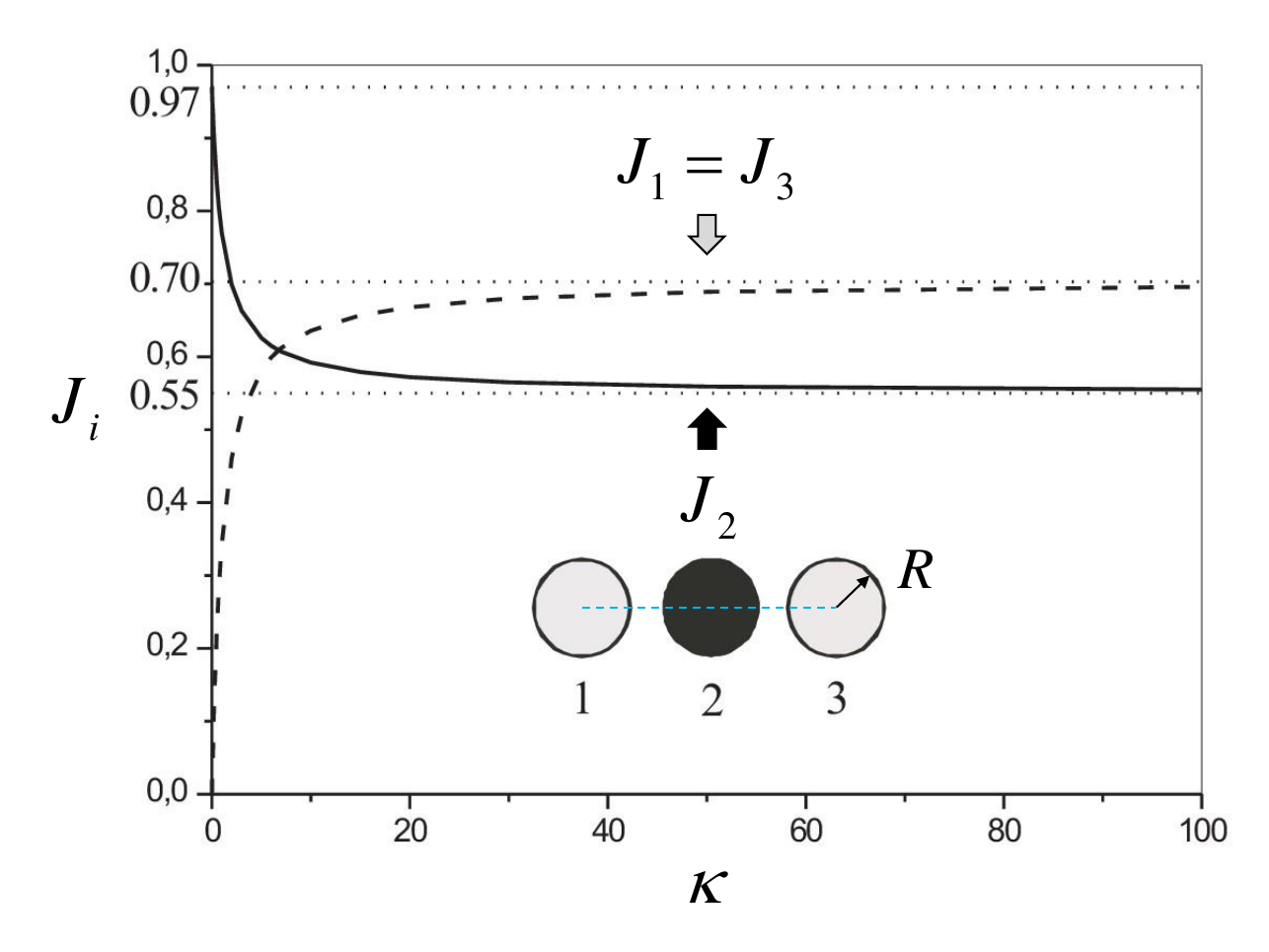}
	\caption{Screening coefficients $ J_i$ as functions of $ \kappa $: long-dashed line indicates screening coefficients $ J_1 = J_3$ for partially reflecting sinks 1 and 3; the solid line indicates  the fully absorbing sink 2 with $ J_2$.}
	\label{fig9}
\end{figure}

\textbf{Fully absorbing sink and two partially reflecting sinks of equal radii.}

Finally, we also give a numerical example for screening coefficients (see Subsection~\ref{subsec:rate}) in the case of three sinks of the same radii ($R_i = R  $, $ i = \overline{1,3} $) equally spaced at spacing $ a $ along their line of centers. Since the sinks numbering system is arbitrary we number sinks from left to right (see Fig.~\ref{fig9}). Moreover we suppose that sinks possess different reactivities, namely, sink 2 is fully absorbing ($\kappa_2 \to \infty$) while the other two sinks are of the same, but varied reactivity ($\kappa_1 = \kappa_3 =\kappa \in \left[0, \infty\right) $). Thus, the reactivities set is simplified to $ \kappa^{(3)}:=\left\{ \kappa, +\infty, \kappa \right\}$. Introduce the notations
$ J\left(\kappa\right):= J_1\left(\kappa^{(3)}\right)=J_3\left(\kappa^{(3)}\right) $ and $ J_2\left(\kappa\right):= J_2\left(\kappa^{(3)}\right) $ and assume that $ a/R = 3 $. 

Fig.~\ref{fig9} shows that in both sinks 1 (3) and sink 2 cases influence of diffusive interaction slighty decreases with increasing intrinsic rate constant $ \kappa $. Furthermore, screening coefficient $ J\left(\kappa\right) $ ($ J_2\left(\kappa\right) $) is monotonically increasing  (decreasing)
approaches corresponding horizontal asyymptotes 
\begin{align}
	\lim_{\kappa \rightarrow +\infty } J\left(\kappa \right) \approx 0.70  \,, \quad  
	\lim_{\kappa \rightarrow +\infty } J_2\left(\kappa \right) \approx 0.55 \,. \label{ex1}
\end{align}
In Fig.~\ref{fig9} short-dashed straight line corresponds to the horizontal asymptotes (\ref{ex1}) for sinks 1 (3) and sink 2, repectively.

For the fully reflecting obstacles limit we have (see corresponding short-dashed straight line)
$$ \lim_{\kappa \rightarrow 0+ } J\left(\kappa \right) = 0  \,, \quad  
\lim_{\kappa \rightarrow 0+ } J_2\left(\kappa \right) \approx 0.97 \,. $$
Clearly, since sink 2 is hidden from both sides by neighbors, even inert obstacles reduce the screening coefficient.

Calculations were carried out with by means of the reduction method~\cite{Kantorovich,Akilov}. Note in passing that, although for the microstructure under consideration when $\kappa \to \infty $ relation (\ref{14e}) gives $ M_3 = 1 $,  the simple-iteration method also converges.

\section{Discussion}\label{sec:discussion}

This section contains discussion of some distinctive characteristics of the method under consideration and some important conclusions of its application for description of the diffusive interaction between sinks. 

First point out that regardless of its technical implementation the GMSV belongs to the class of methods, when corresponding approximate solution satisfies the diffusion equation (\ref{di1}) and regularity at infinity (\ref{di3}), while the boundary conditions (\ref{di2}) are obeyed only approximately. This fact is of primary concern to theoretical modelling of the diffusion-influenced reactions in an arbitrary ensembles of $ N $ sinks. Here, it is worth noting that, e.g., performing solution of the exterior elliptic boundary value problems in unbounded domains by grid methods one faces serious difficulties concerning condition at infinity.


\subsection{Advantages of the ICT technique}\label{subsec:discussionAdvan}

We have solved the problem on diffusive interaction between spherical sinks by means of local Cartesian coordinate systems, rather than by introducing polar spherical coordinates. As it might be appeared at first sight, Cartesian coordinates serve well to solve diffusion problems for rectangular geometries, and not for the problems containing spherical domains.
However, it turns out that above assertion generally is not a case. In connection with this it is significant that Hinsen and Felderhof claimed the following: "... the spherical harmonics have disadvantages in numerical calculations." and below they continued: "in both analytical and numerical calculations the Cartesian moments are often more convenient"~\cite{Felderhof92}.

It appears that applying of irreducible Cartesian tensors technique leads to rather fast convergence of the local concentration expansion (\ref{di6u}). Besides, this approach, contrary to standard one based on the irregular solid spherical harmonics $ {\psi}^-_n \left( r_i, \theta_i , \phi_i \right) $~\cite{Galanti16a,Grebenkov19,Piazza19}, allows us to perform all numerical calculations for the real-valued functions only. Note that real-valued solid harmonics have a number of advantages over often used the complex-valued ones.
The choice of irregular irreducible Cartesian tensors $ \bm{X}_{n}^{-}\left(\mathbf{r}_i \right) $ (\ref{it1vd1}) as basis of functions for solutions of the exterior steady-state diffusion problems in the configuration manifold $  \Omega^- $, implies a very rapid convergence of the numerical solution, as confirmed here with several examples. In addition, contrary to the spherical solid harmonics, there are relatively simple recurcion formulas to calculate analytically the irreducible Cartesian tensors (see Appendix~\ref{subsec:appendixEfim}). Thus, our result is in complete agreement with above Hinsen-Felderhof statement.

Moreover, from pure mathematical viewpoint curvilinear coordinates have one more serious drawback.
The preferred usage of coordinates corresponding to the sink surface is a matter of common knowledge~\cite{Traytak18}. In other words, for a given $ i $-th sink $ \Omega_i $ it is expedient to choose an appropriate local curvilinear coordinates $ \lbrace O_i; {\bm \zeta} \rbrace \subset {\mathbb R}^3 $, where $ {\bm \zeta} = (\zeta_{1}, \zeta_{2}, \zeta_{3}) $. Therefore, some coordinate surface of this coordinate system, say at a fixed index $  m_i \in \overline{1,3} $ there exists a real constant $ c_{i} $ such that the surface of $ i $-th sink reads:
$ \partial \Omega_i = \lbrace {\bm \zeta} \in {\mathbb R}^3 : \zeta_{m_i}= c_{i} \rbrace $.
It is clear that in the Cartesian coordinate system the spherical geometry of the boundary leads to a rather complex formulation of the boundary conditions. So, it seems quite reasonable to use appropriate local spherical coordinates $ \lbrace O_i; r_i, \theta_i , \phi_i \rbrace $ to describe diffusion of $B$-reactants towards $ i $-th spherical sink in domain $ \Omega_i^-:= {\mathbb R}^3\backslash \overline{\Omega }_i $. However, the regularity condition for coordinates to be curvilinear is broken at some points. Particularly, in above case of polar spherical coordinates to have one-to-one property in $ \Omega_i^- $, $Ox_3$ axis and half-plane $ \lbrace x_1 \geq 0, x_2 =0 \rbrace $ should be excluded~\cite{Fomenko09}. A number of other advantages of the ICT technique are described also elsewhere~\cite{Coope65,Hess15,Snider18}.

Finally note that from the analytical point of view, obtained in Sec.~\ref{sec:solution} solution $u_N \left( P \right)$ keeps an explicit analytical dependence on $P$. In turn, from the numerical standpoint, this property allows one to truncate the resolving ISLAE at lower sizes as compared to other numerical methods such as, e.g., a finite elements method~\cite{McCammon13}.

\subsection{Description of the diffusive interaction}

Let us dwell now on several particularly important points on the diffusive interaction which directly follows from our results.

If an array comprises only fully reflecting obstacles we have no diffusive interaction between them. On the other hand in situation when a given array contains at least one absorbing sink the diffusive interaction exists and we should also take into account influence of inert obstacles. 

It is very important to note that here the ICT technique was applied to the diffusion problems under consideration posed in the unbounded domains, whereas the same approach is also applicable to the similar interior boundary value problems describing reactions occurring inside a spherical region (see Refs.~\cite{Galanti16a,Grebenkov19,Piazza19} for details). So, the GMSV in terms of ICT can be very useful not only in the theory of diffusion-influenced reactions, but also in biology, when reactions can take place into bounded domains~\cite{Vazquez10}. We note, incidentally, that Smoluchowski's rate constant (\ref{Zv01}) cannot describe even one-sink reaction inside a larger spherical domain since it does not consider diffusive interaction effect due to the boundary of englobed domain.

In Ref.~\cite{Ovchinnikov89} authors concluded that their system of equations (3.97) solves the $N$-sink diffusion problem under Smoluchowski boundary conditions exactly. However, results obtained in  Sec.~\ref{sec:solution} show that the system (3.97) has been derived for the monopole approximation and, therefore, the above conclusion is incorrect.

It is interesting to note that, clustering of sinks was found to reduce the absorbing rate even more significantly compared with nonclustered systems~\cite{Richards86}. The method gives us possibility to tackle trapping reactions in several clusters of obstacles and sinks in the same way as for separated sinks.

Taking into account our previous results~\cite{Traytak92,Traytak13} and Proposition~\ref{CC2} we can conclude that diffusive interaction "glues" fully absorbing sinks effectively onto a fully absorbing effective sink. Therefore the whole $N$-sink array behaves like a single, isolated fully absorbing sink of some characteristic size. 


For a number of applied tasks, one should use the derived microscopic solution (\ref{mr1}) (or (\ref{di6u})) along with microscopic trapping rates (\ref{d11}) calculating appropriate configurationally averaged values to obtain effective macroscopic values associated with heterogeneous media comprising different types of sinks~\cite{Brailsford76,Loewenberg88,Borodin94,Torquato02}.

Finally, let us dwell on one more possible application of the approach at issue. In Ref.~\cite{Traytak19} we have established for the first time that even classical passive phoresis of microparticles (i.e., motion due to a constant gradient of some scalar field) should be treated by means of their interaction with so-called surrounding bodies, which are responcible for the external gradient. In this context, it is worth to note that the first and only attempt to apply a similar approach for studing particle interactions for thermophoresis was undertaken by Keh and Chen in Ref.~\cite{Chen95}. They used so-called polyadic tensors but in fact, these tensors are the same as the ICT up to a constant factors. Again the GMSV in terms of ICT can be used to investigate these quite a challenging problems.

\section{Conclusions and Perspectives}\label{sec:conclusion}

%
%
%
%
%
%
%
%
%
%

%
%
%
%
%
%
%
%
%
%
%
%

Motivated through new applications in physics, chemistry and biology this paper deals with the generalized method of separation of variables by means of the irreducible Cartesian tensors formalism extending our previous study on the diffusion-controlled reactions to the general case of diffusion-influenced reactions. We have developed here a  complete modification of the above mentioned method which essentially hinges on a number useful properties of the irreducible Cartesian tensors. Formally we reduced the steady-stae diffusion problem to a set of $N$ Laplace equations with respect to partial solution together with the prescribed $N$ boundary conditions.

In order to determine the place and significance of our research we have presented a rather comprehensive review of literature devoted to the theoretical methods for diffusion-influenced reactions on spherical sinks.
Besides, we paid close attention to the mathematical aspects of the appropriate diffusion problem, namely: (a) rigorous statement of the problem, including terminology clarification, (b) refining solution algorithm and (c) justification of the method convergence. The formulation of the geometric part of the diffusion boundary value problems was performed the most fully than anywhere else. We highlighted the importance of a concept of a smooth manifold for the future theoretical investigations on diffusion-influenced processes occuring in domains with multiply connected boundaries.  We have given a proof of the irregular to regular translation addition theorem in terms of the irreducible Cartesian tensors which seems to be the simplest among previously known theorem for solid harmonics in terms of polar spherical coordinates. For the first time corresponding solution algorithm was divided into eight clear steps. Then we implemented this algorithm step by step in full detail to solve the posed diffusion problem. 
A convergence analysis was conducted to verify the possibility of simple-iteration method usage in order to solve appropriate resolving infinite system of linear algebraic equations. It turns out that the method of reduction may be used to solve numerically the resolving system to any degree of accuracy even when the simple-iteration method is inapplicable. Wherein the desired accuracy of the solution can be achieved just by increasing the truncation degree.
Furthermore, note that the method has advantages such as shorter calculation time and more accurate results compared to purely numerical methods. 

Our results fully confirm known Hinsen-Felderhof statement that it is more preferred to work with harmonic functions using their irreducible Cartesian tensor form throughout diffusion-reaction problems even for ensembles of spherical sinks.

This research provides also a better understanding of the reactants diffusion among arrays comprising arbitrary collection of arbitrary many fully absorbing, reflecting and partially reflecting sinks from both mathematical and physical viewpoints.

The approach developed here is applicable to a wide class of reaction-diffusion systems, providing useful tests for numerical calculations along with approximate analytical estimates of the local concentration field and reaction rates in many specific cases of practical importance. 


Finally it is important to highlight that numerical methods can handle diffusion problems for complex geometries in unbounded domains, yet they require to introduce an articial boundary, reducing the original exterior domain to a bounded domain. The generalized method of separation of variables does not have this drawback and provides simple direct numerical and analytical calculations.
So we can obtain approximations as the truncated point multipole expansions up to the monopole, dipole, quadrupole etc. order. 

It is essential to note that self-phoresis of the Janus particles moving due to a chemical	reaction in the surrounding fluid is investigated by means of monopole approximation~\cite{Stark18}.
Using the approach considered here this important problem may be attack taking into account multipole corrections~\cite{Traytak19}.

The last but not least: in the present study we clearly show that the theory of 3D differentiable manifolds is an appropriate mathematical tool to describe diffusion and reactions in domains with multiply connected boundaries. In this context, we note that, rigorously posing and solving the appropriate diffusion problems, we, nevertheless, endeavored to avoid considering a number of inherent subtle mathematical points. This issues will be a topic for the future work.

\renewcommand{\theequation}{A.\arabic{equation}}  
\setcounter{equation}{0}  
\section*{Appendix. Background on the irreducible Cartesian tensors}
\label{sec:appendix}

To make this paper maximally self-contained and facilitate applications in diffusion-influenced reactions theory, we provide here necessary mathematical notations, definitions and facts, required in rigorous formulation and applications of the ICT technique.

\subsection{Definitions and examples}\label{sec:appendixDef}

It is well-known that there exits two alternative approaches to harmonic functions theory based on either Laplace's representation of solid harmonics with respect to the polar spherical coordinates or Maxwell's theory of poles given in terms of the successive derivatives of fundamental solution written for a Cartesian coordinate system~\cite{Courant53,Marchenko}. The latter approach may be formulated with the aid of the irreducible Cartesian tensors.

Generally speaking, the irreducible tensors transform according to irreducible representation of the special rotation group  $SO\left( 3\right)$. In other words the irreducible tensors like spherical harmonics form a basis in the representation space of group $SO\left( 3\right)$ (see~\ref{subsec:appendixCon}).	
To our knowledge, the ICT simple mathematical description suitable for applied physicists has been provided in Ref.~\cite{Coope65}. Moreover, for many years there has been a significant number of publications devoting to the use of the ICT formalizm in a wide variety of applications (see, e.g., Refs.~\cite{Brenner64,Anderson74,Cipriani82,Martynov98} and references therein).
As regards adopted notations and definitions, in this paper we basically follow known books by Hess~\cite{Hess15} and Snider~\cite{Snider18}, where  comprehensive discussions on the various facets of irreducible Cartesian tensors are presented.

A real tensor field of $ n $-th order defined on $ {\mathbb R}^3 $ we denote by $\bm{T}_n \left(\mathbf r \right)$ in
coordinate-free notation. The same tensor field can be also described by its components with respect to a fixed orthonormal basis  $  \lbrace {\mathbf e}_{\alpha} \rbrace_{\alpha = 1}^{3} $ given in $ \lbrace O; \mathbf r \rbrace $, i.e. $ {\bm T}_n \left({\mathbf e}_{\gamma _1}, \ldots , {\mathbf e}_{\gamma _n} \right)  = T_{{\gamma _1} \ldots {\gamma _n}}\left(\mathbf r \right) $. Here and in the following, $ T_{{\gamma _1} \ldots {\gamma _n}}\left(\mathbf r \right) $ are $ 3^n $ components of $ n $-rank Cartesian tensor field $\bm{T}_n \left(\mathbf r \right)$. It is clear that the action of the nabla operator with $ \partial _{r_{\mu}} := \nabla _{\mu} $
to a field, which is a tensor of rank $ l $, yields a tensor of rank $l+1  $. We recall in passing that one should not distinguish between co- and contravariant index due to duality of Cartesian basis~\cite{Snider18}.

The product symbol $ \odot^q $ represents a {\it full tensorial contraction of multiplicity $q$} or shortly the {\it $q$-fold contraction} between two tensors with the convention $ {\bm S}_0 \odot {\bm T}_{0} :=  {\bm T}_{0} $. Particularly, for tensors $ {\bm S}_n $ and $ {\bm T}_{m} $ and $ q \leq n, m $ we have a tensor 
$$ {\bm S}_n \odot^q {\bm T}_{m} := \sum_{\gamma _1} \ldots \sum_{\gamma _q} S _{\gamma _1 \ldots \gamma _n}\odot^q T_{\mu _1 \ldots \mu _m}\,  $$
of rank $n+m-2q $, whereas it is a scalar if $ n=m=q $.

Any tensor, which cannot be reduced to a tensor of lower rank is known as an {\it irreducible tensor}~\cite{Hess15}.
While a general (reducible) $ n$-th order Cartesian tensor $ {\bm T}_n $ has $ 3^n $ components, corresponding irreducible tensor besides the rank $ n $ is likewise characterized by its weight $ j $~\cite{Coope65}. It is common knowledge that latter tensor has $ 2j+1 $ independent components, which form the basis of the $ j $ weight irreducible representation of the full 3D special rotation group $ SO \left( 3\right)$ including mirror~\cite{Snider18}. 

Furthermore, among all ICT obtained from a given $ n $ rank tensor by means of reductions there exists one unique irreducible tensor with the same weight and rank $ n=j $. It is known in literature as ICT {\it in natural form}~\cite{Coope65}. We deal with irreducible tensors in natural form because they best suited to apply within the scope of the GMSV.

Thus, for $n$ rank Cartesian tensors our focus lies on their irreducible parts in natural form at that we use notations $\mathop {{\bm T}_n }\limits^{%
	\rule[-.04in]{.01in}{.05in}\rule{.15in}{.01in}\rule[-0.04in]{.01in
	}{.05in}}$ and $\mathop {T_{{\gamma _1} \ldots {\gamma _n}} }%
\limits^{\rule[-.04in]{.01in}{.05in}\rule{.34in}{.01in}%
	\rule[-0.04in]{.01in}{.05in}} $, where symbol $  \mathop {\left(\: ... \:\right)}%
\limits^{\rule[-.04in]{.01in}{.05in}\rule{.25in}{.01in}%
	\rule[-0.04in]{.01in}{.05in}} $ stands for the 
irreducible part of a Cartesian tensor $ {\bm T}_n $ of rank $ n \geq 2 $  with respect to the special rotation group $SO(3)$~\cite{Hess15,Snider18}.

The irreducible tensor in natural form is appeared to be fully symmetric and traceless, therefore, this allows us to use rather simple definition of the irreducible Cartesian tensors~\cite{Junk12,Mane16}.
\begin{definition}\label{ICT1}
	A Cartesian tensor of rank $  n \geq 2 $ is called irreducible Cartesian tensor, if it is: (i) totally symmetric under an arbitrary permutation of the indices; and (ii) traceless under the contraction of any pair of indices.
\end{definition}

An important point is that according to definition (\ref{ICT1}) the ICT can be expressed in terms of successive partial derivatives of the generating function $ 1/r $ for all points in $\mathbb{R}^3 \backslash \left\{ {\mathbf 0} \right\}$.	
\begin{proposition} \label{ICT2}
	An irreducible Cartesian tensors of rank $n$ (when $ n \geq 2 $)  may be defined as~\cite{Courant53,Hess15} 
	\begin{equation}
		\mathop {\hspace{0.2cm}{r}_{\gamma _1}{\ldots r}_{\gamma _n}}\limits^{%
			\rule[-.04in]{.01in}{.05in}\rule{.45in}{.01in}\rule[-0.04in]{.01in
			}{.05in}}=\frac{\left( -1\right) ^n}{\left( 2n-1\right) !!}r^{2n+1}\partial
		_{r_{\gamma _1}} \ldots \partial _{r_{\gamma _n}}\left( \frac 1r\right)  \,.
		\label{it1v}
	\end{equation}
	Here $r_{\gamma _\nu }$ are Cartesian coordinates of a vector $\mathbf{r}$ ($\gamma _\nu =\overline{1,3}$),
	and $\left( 2n-1\right) !!:=1\cdot 3\cdot \ldots \cdot \left( 2n-1\right) $ indicates the double factorial for odd natural numbers with the convention $ (-1)!!:=1 $.
\end{proposition}
\begin{proof}
	Check that the tensors introduced by Eq. (\ref{it1v}) are symmetric for any pair of indexes and traceless.
	
	(i)  Indeed,  Eq. (\ref{it1v}) involves directly that the introduced tensor is symmetric with respect to permutation of its indexes 
	\begin{equation}
		\mathop {\hspace{0.2cm}{r}_{\gamma _1} \ldots {r}_{\alpha} \ldots {r}_{\beta}\ldots {r}_{\gamma _n}}\limits^{%
			\rule[-.04in]{.01in}{.05in}\rule{1.15in}{.01in}\rule[-0.04in]{.01in
			}{.05in}} = \mathop {\hspace{0.2cm}{r}_{\gamma _1} \ldots {r}_{\beta} \ldots {r}_{\alpha}\ldots {r}_{\gamma _n}}\limits^{%
			\rule[-.04in]{.01in}{.05in}\rule{1.15in}{.01in}\rule[-0.04in]{.01in
			}{.05in}}\,.	\label{it1vc}
	\end{equation}
	(ii) Moreover, due to harmonicity of the generating function 		
	\begin{equation}
		{\bm{\nabla}}^2 \left( \frac{1}{r} \right)  = \delta_{\alpha {\beta}} \odot^2 \partial_{r_{\alpha}}\partial_{r_{\beta}} \left( \frac{1}{r} \right)= 0 \quad \mbox{for} \quad  r>0\,  \label{int3x}
	\end{equation}
	tensors (\ref{it1v}) are also traceless with respect to any pair of their indexes
	\begin{equation}
		\delta_{\alpha {\beta}} \odot^2  \mathop {\hspace{0.2cm}{r}_{\gamma _1} \ldots {r}_{\alpha} \ldots {r}_{\beta}\ldots {r}_{\gamma _n}}\limits^{%
			\rule[-.04in]{.01in}{.05in}\rule{1.15in}{.01in}\rule[-0.04in]{.01in
			}{.05in}} = 0 \, , 	\label{it1vd}
	\end{equation}
	where $ \delta_{\alpha {\beta}} $ is the Kronecker delta. Thus Eq. (\ref{it1v}) satisfies both points of definition (\ref{ICT1}).
\end{proof}
\begin{remark}
	Note furthermore that traces cannot be defined for any scalars $ a_0 \in {\mathbb R} $ (tensors of rank $ 0 $) or vectors $ \left({r}_{1}, {r}_{2}, {r}_{3}\right) $ (tensors of rank $ 1 $)). Nethertheless they are harmonic and, therefore, irreducible by definition~\cite{Hess15}.  
\end{remark}

In applications are also used the {\it multipole potentials} defined by: $ \bm{X}_{n}^{ \mp }\left(\mathbf{r} \right)={X}^{ \mp } _{\gamma _1 \ldots \gamma _n}\left( {\bf r}\right) $~\cite{Hess15}, where components read
\begin{align}
	{X}^- _{\gamma _1 \ldots \gamma _n}\left( {\bf r}\right):={\left( -1\right) ^n}\partial
	_{r_{\gamma _1}} \ldots \partial _{r_{\gamma _n}}\left( \frac 1r\right)\,,  \label{it1vd1}\\
	{X}^+ _{\gamma _1 \ldots \gamma _n}\left( {\bf r}\right):= \left( 2n-1\right)!!\mathop {\hspace{0.2cm}{r}_{\gamma _1}{\ldots r}_{\gamma _n}}\limits^{%
		\rule[-.04in]{.01in}{.05in}\rule{.45in}{.01in}\rule[-0.04in]{.01in
		}{.05in}} \,.\label{it1vd2}
\end{align}
Owing to definition  they are connected by the relation
\begin{align}
	\bm{X}_{n}^{+}\left(\mathbf{r} \right) =  r^{(2n+1)}\bm{X}_{n}^{-}\left(\mathbf{r} \right) \,. 
	\label{it1vd2a}
\end{align}
For reasons explained below (see~\ref{subsec:appendixCon}), it is quite natural to call potentials $ \bm{X}_{n}^{-} $  and  $ \bm{X}_{n}^{+} $ {\it irregular} and {\it regular ICT}, respectively~\cite{Felderhof92}. However, throughout this paper we use ICT defined by Eq. (\ref{it1v}), since they are most extensively employed in applications.

Denote by $  d^2 {\hat{r}}$ the magnitude of the surface
element on the unit sphere in global coordinates 
\begin{align}
	\partial \Omega_1 \equiv \partial \Omega \left({\bm 0}; 1 \right) := \lbrace {\mathbf r} \in \mathbb{R}^3: {\hat r}=1 \rbrace \,.
	\label{it1vd2b}
\end{align}
Then for all $ n \geq 0 $ we can write the following {\it orthogonality condition}~\cite{Hess15}
\begin{eqnarray}
	\oint_{\partial \Omega_1}\mathop {{\hat{r}}_{\gamma _1}\ldots {\hat{r}}_{\gamma _n}}\limits^{%
		\rule[-.04in]{.01in}{.05in}\rule{.47in}{.01in}\rule[-0.04in]{.01in
		}{.05in}} \cdot \mathop {{\hat{r}}_{\mu _1}\ldots {\hat{r}}_{\mu _k}}\limits^{%
		\rule[-.04in]{.01in}{.05in}\rule{.49in}{.01in}\rule[-0.04in]{.01in
		}{.05in}} d^2 {\hat{r}} 
	= \frac{4\pi n!}{\left( 2n+1 \right)!!}\delta_{n k} \Delta^{(n)} _{\gamma _1 \ldots \gamma _n, \mu
		_1  \ldots \mu _n}\,, \label{Asur2}
\end{eqnarray}
where in case of sphere $ \hat{\mathbf n} = \hat{\mathbf r} $  (in component notation $ \hat{r}_{\alpha}=r_{\alpha}/r $) is a unit vector identifying a point on $ \partial \Omega_1$ and $ \left. \mathop {{r}_{\gamma _1}\ldots {{r}}_{\gamma _n}}\limits^{%
	\rule[-.04in]{.01in}{.05in}\rule{.48in}{.01in}\rule[-0.04in]{.01in
	}{.05in}} \right| _{\partial \Omega_1} = \mathop {{\hat{r}}_{\gamma _1}\ldots {\hat{r}}_{\gamma _n}}\limits^{%
	\rule[-.04in]{.01in}{.05in}\rule{.47in}{.01in}\rule[-0.04in]{.01in
	}{.05in}}$ is the restriction of the ICT to the unit sphere (\ref{it1vd2b}). In Eq. (\ref{Asur2})
$ \Delta^{(n)} _{\gamma _1 \ldots \gamma _n, \mu_1  \ldots \mu _k}$ stands for the $ 2n $ rank projector, which project any tensor of rank $ n $ into its irreducible part. Particularly, three first projectors are~\cite{Hess15}
\begin{align*}
	\Delta^{(0)} = 1\,, \quad  \Delta^{(1)} = \delta_{\alpha {\beta}} \,, \\
	\Delta^{(2)} _{\alpha \beta, \alpha' \beta'} = \frac{1}{2} \left( \delta_{\alpha \alpha'}\delta_{\beta \beta'}+\delta_{\alpha \beta'}\delta_{\beta \alpha'} \right) - \frac{1}{3}\delta_{\alpha \beta}\delta_{\alpha' \beta'}\,.
\end{align*}

\subsection{Efimov's recurrent formula}\label{subsec:appendixEfim}
In~\ref{subsec:discussionAdvan} we have already mentioned that an important advantage of the ICT over solid harmonics is relatively simple recurcion formulas obtained by rather easy successive differentiation of generating function. Interested readers are encouraged to refer to works, where explicit recurrent expressions for the ICT are comprehensively presented~\cite{Cipriani82,Blaive84,Applequist89,Felderhof92,Ehrentraut98,Hess15,Snider18}. Here let us dwell briefly on a little-known, but simple and powerful approach proposed by Efimov~\cite{Efimov79,Efimov22} that allowed us to calculate ICT. 

To derive the relevant Efimov's recurrent formula first we introduce an auxiliary tensor field
\begin{align}
	\Gamma _{\gamma
		_1 ... \gamma _n}\left( {\bf r}\right):= {\left( -1\right) ^n}{X}^+ _{\gamma _1 \ldots \gamma _n}\left( {\bf r}\right)\,.  \label{Ef1}
\end{align}

Inversion of the radius vector $ \mathbf{r} \rightarrow \mathbf{r} r^{-2} $ leads to an isomorphic  mapping  of the  gradient operator in the inverted space of harmonic functions (with opposite sign), i.e. the vector operator $\widehat{\mathbf{D}}=-r^{-1} {\nabla _{1/r}} r$. Direct calculations yield the expression
\begin{equation}
	\widehat{\mathbf{D}}:=2\mathbf{r}\left( \mathbf{r}\cdot \nabla \right)
	-r^2\nabla +\mathbf{r}\,. \label{Ef2}
\end{equation}
With the help of this operator it may be shown that the following recursion holds 
\[
\Gamma _0=1, \quad \Gamma _{\gamma _n \gamma _1...\gamma _{n-1}}^{\left( n\right) }=\widehat{D}_{\gamma _n} \cdot \Gamma
_{\gamma _1...\gamma _{n-1}}^{\left( n-1\right) } \quad
\mbox{for} \quad n\in \mathbb{N} . 
\]
Taking this relations into account one easily derives the desired Efimov's formula~\cite{Efimov79}
\begin{align}
	\Gamma _{\gamma _1 ...\gamma _n}=\widehat{D}_{\gamma _1}\widehat{D}_{\gamma _2}...\widehat{D}%
	_{\gamma _n}\cdot 1\, . \label{Ef3}
\end{align}

Hence, in the explicit form the first several irreducible Cartesian tensors of ranks $ n $ corresponding to the monopole ($n=0$), dipole ($n=1$), quadrupole ($n=2$) and octupole ($n=3$) potentials read as follows:
\begin{align*}
	\mathop {\hspace{0.2cm}a_0}\limits^{%
		\rule[-.04in]{.01in}{.05in}\rule{.13in}{.01in}\rule[-0.04in]{.01in
		}{.05in}} := a_0  \,, 
	\mathop {\hspace{0.2cm}{r}_{\gamma _1}}\limits^{%
		\rule[-.04in]{.01in}{.05in}\rule{.13in}{.01in}\rule[-0.04in]{.01in
		}{.05in}} := {r}_{\gamma _1}  \,, 
	\mathop {\hspace{0.2cm}{r}_{\gamma _1}{r}_{\gamma_2}}\limits^{%
		\rule[-.04in]{.01in}{.05in}\rule{.3in}{.01in}\rule[-0.04in]{.01in
		}{.05in}} = {r}_{\gamma _1}{r}_{\gamma _2}  - \frac{1}{3}r^2 \delta_{\gamma _1 {\gamma _2}}\,, \\
	\mathop {\hspace{0.2cm}{r}_{\gamma _1}{r}_{\gamma _2}{r}_{\gamma _3}}\limits^{%
		\rule[-.04in]{.01in}{.05in}\rule{.43in}{.01in}\rule[-0.04in]{.01in
		}{.05in}} = {r}_{\gamma _1}{r}_{\gamma _2}{r}_{\gamma _3} - \frac{1}{5}r^2 \left({r}_{\gamma _1}\delta_{\gamma _2 {\gamma _3}}+{r}_{\gamma _2}\delta_{\gamma _1 {\gamma _3}}+{r}_{\gamma _3}\delta_{\gamma _1 {\gamma _2}}\right) \, . \\
\end{align*}


\subsection{Connection with solid harmonics}\label{subsec:appendixCon}

As indicated in Sec.~\ref{sec:intro}, solid harmonics were thoroughly used previously to investigate many-sink effects during diffusion-influenced reactions. So, it is expedient to point to an intimate relationship between them and the ICT~\cite{Stone75,Normand82,Dennis04,Hess15}.

Recall that {\it real solid harmonics} $ \psi_{nm}\left(r, \theta, \phi \right) $, where the integer $ m $ such that $- n \leq m \leq n $ and $n=\overline{0,\infty }$, are harmonic functions in spherical coordinates $ \lbrace O; r, \theta, \phi \rbrace $~\cite{Tikhonov63}. For brevity we shall denote them as $ \psi_{n}\left(r, \theta, \phi \right) $.
Further, according to their behavior at the origin $ O $, two kinds of solid harmonics are distinguished:

(1)  {\it irregular} at $ O $ solid harmonics $ {\psi}_{n}^{-}\left(r, \theta, \phi \right) $ has a pole of order $n+1$ singularity at $ O $;

(2) {\it regular} at $ O $ solid harmonics $ {\psi}_{n}^{+}\left(r, \theta, \phi \right) $ are real-valued harmonic homogeneous polynomials of order $ n $.

Hence, irregular solid harmonics $ {\psi}_{n}^{-} $ automatically satisfies regularity condition at infinity (\ref{di3}) and, therefore, they are regular at infinity. 

Clearly, irregular $ \bm{X}_{n}^{-} $ (\ref{it1vd1}) and regular  $ \bm{X}_{n}^{+} $ (\ref{it1vd2}) ICT correspond to irregular and regular solid harmonics $ {\psi}_{n}^{-} $ and $ {\psi}_{n}^{+} $, respectively.

Represent tensors $ \bm{X}_{n}^{\pm}\left( {\bf r}\right) $ in the form most closely resemble classical irregular $ {\psi}_{n}^{-}\left(r, \theta, \phi \right) $, regular $ {\psi}_{n}^{+}\left(r, \theta, \phi \right) $ solid and real spherical $ {Y}_{n}\left(\theta, \phi \right) $ harmonics:
\begin{align}
	{X}^- _{\gamma _1 \ldots \gamma _n}\left( {\bf r}\right)=r^{-\left(n+1\right)} Y_{{\gamma _1} \ldots {\gamma _n}}\left({\hat{\mathbf r}} \right) \, , \label{SH1} \\
	{X}^+ _{\gamma _1 \ldots \gamma _n}\left( {\bf r}\right)=r^n Y_{{\gamma _1} \ldots {\gamma _n}}\left({\hat{\mathbf r}} \right) \,, \label{SH2}\\
	Y_{{\gamma _1} \ldots {\gamma _n}}\left({\hat{\mathbf r}} \right):=\left( 2n-1\right)!!\mathop {{\hat{r}}_{\gamma _1}\ldots {\hat{r}}_{\gamma _n}}\limits^{%
		\rule[-.04in]{.01in}{.05in}\rule{.45in}{.01in}\rule[-0.04in]{.01in
		}{.05in}}\,. \label{SH3}
\end{align}
It is apparent that the ICT  $ Y_{{\gamma _1} \ldots {\gamma _n}}\left({\hat{\mathbf r}} \right) $ like spherical harmonics depend on the direction of vector $ \mathbf{r} $ only~\cite{Hess15}.

As noted above, the ICT similar to solid harmonics form a basis in the representation space of the group $SO\left( 3\right)$, therefore, there exists a non-singular linear transformation and its inverse in space $\mathbb{R}^{2n+1}$ which map these basises into each other~\cite{Courant53}. Using known {\it Sylvester's theorem}~\cite{Courant53,Dennis04} one can find the desired connection between tensors $ \bm{X}_{n}^{\pm} $ and classical real solid harmonics $ {\psi}_{n}^{\pm} $ explicitly.  We emphasize here that original Sylvester's proof is rather sophisticated~\cite{Courant53}, therefore, interested readers are encouraged to refer to a simple proof of Sylvester's theorem proposed by Backus in Ref.~\cite{Backus70}. 

Thus, it can be shown that Sylvester's theorem implies
\begin{eqnarray}
	\frac{1}{r^{n+1}}Y_{n}\left(\theta, \phi \right) = \sum_{m=0}^{n} \sum_{l=0}^{m}  c_{lm}  \partial_{r_1}^{l} \partial_{r_2}^{m-l} \partial_{r_3}^{n-m} \left( \frac{1}{r} \right)\,, \label{SH4}\\ 
	\text{where} \hspace{0.6em} Y_{n}\left(\theta, \phi \right) = \sum_{m=0}^{n}  Y_{n}^m\left(\theta, \phi \right) \,, \nonumber \\
	Y_{n}^m\left(\theta, \phi \right) :=\left( a^m_n \cos m\phi +  b^m_n \sin m\phi \right)P^m_n \left(\cos \theta \right)\,  \label{SH5}
\end{eqnarray}
is the real spherical harmonic of degree $ n $ and order $ m $.  In Eq. (\ref{SH5}) $  a^m_n $, $ b^m_n $ and $ c_{lm} $ are some known real coefficients and $ P^m_n \left(\cos \theta \right) $ are associated Legendre functions~\cite{James68}.

Particularly, for the axially symmetric case ($ m=0 $), taking all derivatives along $ O r_3 $ axis, we have $ r_{\gamma _1} = r_{\gamma _2}= \ldots = r_{\gamma _n} = r_3$ and formula (\ref{SH4}) boils down to well-known relation~\cite{Courant53}
\begin{align}
	\frac 1{r^{n+1}} P_n\left( \cos \theta \right) = \frac{\left( -1\right) ^n}{n!} {\partial^n _{r_3}\left( \frac 1r\right)}\, . \label{SH6}
\end{align}

\bibliography{ICTrefs}


\end{document}